\documentclass[preprint,nonatbib]{elsarticle}

\usepackage[english]{babel}
\usepackage[acronym, nonumberlist]{glossaries}
\usepackage{amsmath, amsfonts, amssymb, mathtools, amsthm}
\usepackage{graphicx}
	\DeclareGraphicsExtensions{.pdf,.jpeg,.png}
\usepackage{circuitikz}
    \usetikzlibrary{positioning, fit}
\usepackage{xcolor}
	\definecolor{dark_blue}{HTML}{001219}
	\definecolor{blue}{HTML}{005F73}
	\definecolor{red}{HTML}{9B2226}
	\definecolor{orange}{HTML}{BB3E03}
	\definecolor{yellow}{HTML}{EE9B00}
\usepackage{siunitx}
	\DeclareSIUnit\perunit{pu}
	\DeclareSIUnit[number-unit-product = ]\percent{\%}
\usepackage[caption=false, font=normalsize]{subfig}
\usepackage{booktabs}
\usepackage[hyphens]{xurl}
\usepackage[hidelinks]{hyperref}
\usepackage{bm}
\usepackage{dirtytalk}
\usepackage{enumitem}
\usepackage{algorithm}
\usepackage{algpseudocode}
\usepackage{placeins}
\theoremstyle{definition}

\newtheorem{assumption}{Assumption}
\newtheorem{definition}{Definition}
\newtheorem{theorem}{Theorem}
\newtheorem{corollary}{Corollary}[theorem]
\newtheorem{proposition}{Proposition}
\newtheorem{remark}{Remark}

\NewDocumentCommand{\inv}{s m}{%
  \IfBooleanTF{#1}{%
    \left(#2\right)^{-1}%
  }{%
    #2^{-1}%
  }%
}
\NewDocumentCommand{\trans}{s m}{%
  \IfBooleanTF{#1}{%
    \left(#2\right)^{\mathsf{T}}%
  }{%
    #2^{\mathsf{T}}%
  }%
}
\NewDocumentCommand{\herm}{s m}{%
  \IfBooleanTF{#1}{%
    \left(#2\right)^{\mathsf{H}}%
  }{%
    #2^{\mathsf{H}}%
  }%
}

\newcommand{\complex}{\mathrm{j}}
\DeclareMathOperator{\diagop}{diag}
\newcommand{\diag}[1]{\diagop\!\left(#1\right)}
\DeclareMathOperator{\reop}{Re}
\newcommand{\re}[1]{\reop\!\left(#1\right)}
\DeclareMathOperator{\imop}{Im}
\newcommand{\im}[1]{\imop\!\left(#1\right)}
\newcommand{\expectation}[1]{\mathbb{E}\!\left[#1\right]}
\newcommand{\kernel}[1]{\ker\!\left(#1\right)}
\DeclareMathOperator{\rangeop}{range}
\newcommand{\range}[1]{\rangeop\!\left(#1\right)}
\DeclareMathOperator{\rankop}{rank}
\newcommand{\rank}[1]{\rankop\!\left(#1\right)}
\DeclareMathOperator{\spanMathop}{span}
\newcommand{\spanMath}[1]{\spanMathop\left(#1\right)}
\DeclareMathOperator{\traceop}{tr}
\newcommand{\trace}[1]{\traceop\!\left(#1\right)}

\DeclareMathOperator{\covop}{cov}
\newcommand{\cov}[1]{\covop\!\left(#1\right)}
\DeclareMathOperator{\projop}{proj}
\newcommand{\proj}[2]{\projop_{#1}\!\left(#2\right)}
\DeclareMathOperator*{\argmin}{argmin}
\DeclareMathOperator{\maxMathop}{max}
\newcommand{\maxMath}[1]{\maxMathop\!\left(#1\right)}

\myfooter[L]{\makebox[\textwidth][l]{Preprint\hfill\today}}
\newacronym{siso}{SISO}{single-input single-output}
\newacronym{miso}{MISO}{multiple-input single-output}
\newacronym{mimo}{MIMO}{multiple-input multiple-output}
\newacronym[shortplural=FCRs, longplural=frequency containment reserves]{fcr}{FCR}{frequency containment reserve}
\newacronym[shortplural=DFTs, longplural=discrete Fourier transforms]{dft}{DFT}{discrete Fourier transform}
\newacronym[shortplural=TFs, longplural=transfer functions]{tf}{TF}{transfer function}
\newacronym[shortplural=GPs, longplural=Gaussian processes]{gp}{GP}{Gaussian process}
\newacronym[shortplural=GPRs, longplural=Gaussian process regressions]{gpr}{GPR}{Gaussian process regression}
\newacronym{gptf}{GPTF}{Gaussian process transfer function}
\newacronym{miso-gptf}{MISO-GPTF}{multiple-input single-output Gaussian process transfer function}
\newacronym{map}{MAP}{maximum a posteriori}
\newacronym[shortplural=SVDs, longplural=singular value decompositions]{svd}{SVD}{singular value decomposition}
\newacronym[shortplural=FIRs, longplural=finite impulse responses]{fir}{FIR}{finite impulse response}
\newacronym{mc}{MC}{Monte Carlo}
\newacronym[shortplural=SNRs, longplural=signal-to-noise ratios]{snr}{SNR}{signal-to-noise ratio}

\begin{document}
	\begin{frontmatter}
	\title{Estimating Closed-Loop Multiple-Input Single-Output Dynamic Systems using Gaussian Processes}

	\author[ntnu,statkraft]{Erik-Anant Stedjan Narayan\corref{cor1}}
	\cortext[cor1]{Corresponding author: erik-anant.s.narayan@ntnu.no}
	\author[sintef]{Sigurd Hofsmo Jakobsen}
	\author[ntnu]{Kjetil Obstfelder Uhlen}

	\affiliation[ntnu]{organization={Department of Electric Energy, Norwegian University of Science and Technology}, city={Trondheim}, country={Norway}}       
	\affiliation[statkraft]{organization={Statkraft Energi AS}, city={Oslo}, country={Norway}}     
	\affiliation[sintef]{organization={SINTEF Energy Research}, city={Trondheim}, country={Norway}}

	\begin{abstract}
		This paper extends complex \acrlong{gpr}-based \gls{tf} estimation from the \acrlong{siso} setting to closed-loop \gls{miso} systems. Accurate estimation of individual \glspl{tf} in closed-loop \gls{miso} systems is challenging as feedback and correlation between measured signals can obscure the contributions of each input. We formulate a \acrlong{gp}-based estimator for \gls{miso} systems, provide analytical conditions under which each system module can be separately identified, and show that the estimator converges to the true \glspl{tf} in the zero noise limit. Moreover, we demonstrate how conventional closed-loop bias handling techniques relying on external excitation can be employed to construct appropriate regressors. \acrlong{mc} simulations support the analytical identifiability results and highlight the trade-off between closed-loop bias and error introduced by constructing exogenous regressors partly from endogenous signals. The presented work offers an expressive, nonparametric alternative to existing estimation approaches for \gls{miso} systems, while providing a probabilistic characterization of the estimation uncertainty.
	\end{abstract}

\end{frontmatter}

\glsresetall

	\section{Introduction}\label{sec:introduction}
This work is concerned with the challenge of identifying multiple system components in closed-loop \gls{miso} systems from measurements of noisy outputs and possibly correlated inputs. Developing models of such systems from measurements is central in fields such as automotive and aerospace engineering and industrial process automation, but the main motivating application is here power system management.

The ongoing shift towards a more renewable-based power system is fundamentally changing the dynamics of the system~\cite{MilanoEtAl2018, GuGreen2023}. Conventional synchronous generators are increasingly being replaced or supplemented by converter-in\-ter\-faced generation. The increasing diversity in generation and load types, control strategies, and operating conditions challenge the reliance on existing physics-based modeling, driving a need for improved system estimation techniques~\cite{ENTSO-E2025}. While \gls{siso} approaches can be appropriate~~\cite{SaarinenEtAl2015}, several relevant identification problems are inherently closed-loop and \gls{miso}. Conventional system identification techniques have previously been employed to estimate the frequency control dynamics of synchronous generators~\cite{JakobsenEtAl2020}, demonstrating both the practical relevance of a \gls{miso} formulation and the feasibility of identifying relevant dynamics in closed-loop operation.

The challenges associated with identifying power system dynamics hint at a more fundamental and general aspect of estimation, namely, the interplay between data and model structure. One way of categorizing system identification methods is into parametric and nonparametric approaches. Classical parametric methods build models from discrete model structures, independent of the data size~\cite{Ljung1999}. Nonparametric approaches, on the other hand, construct flexible models that adapt to the data size.

Within nonparametric estimation, the subgenre of kernel-based regularization methods have in the last decade gained increasing popularity within the system identification community~\cite{PillonettoEtAl2014}---even to the point where some argue that it has ushered in a new paradigm~\cite{LjungandMu2020}. Common for these approaches is that they frame system identification as a function estimation problem in some possibly infinite-dimensional space, and then handle the arising ill-posedness of such an endeavor by employing regularization methods. These regularization methods admit a Bayesian interpretation, and so the function estimation problem can be recast in the framework of \gls{gpr}~\cite{RasmussenWilliams2006}. Doing so offers some inherent advantages, e.g.,
\begin{enumerate}[label=(\arabic*)]
    \item the choice of covariance functions lets one incorporate prior structural system information like smoothness and stability;
    \item the resulting nonparametric representation does not require a model order to be specified a priori; and
    \item the approach results in a closed-form posterior distribution with both pointwise estimates and quantified uncertainty bounds.
\end{enumerate}

Time-domain methods leveraging this were introduced in~\cite{PillonettoDeNicolao2010, PillonettoEtAl2011}, and later extended to give continuous frequency-domain models using complex \gls{gpr} together with suitable covariance functions~\cite{LataireChen2016, HallemansEtAl2022}. The resulting method is aptly named the \gls{gptf} estimator. 

A key characteristic of this frequency-domain estimator is the appearance and handling of a transient term when estimating from finite signal recordings. The \gls{tf} and transient are modeled as independent \glspl{gp} and are estimated simultaneously from the posterior conditional mean. This treatment helps define the posterior decomposition, but it does not itself establish physical or structural identifiability. 

This issue becomes more important when we not only seek to estimate and distinguish a single module from a transient, but multiple modules in a \gls{miso} system. Classical system identification has long recognized that unique module recovery in \gls{miso} systems requires sufficiently informative excitation, commonly expressed through persistent excitation or nonsingularity of a regressor Gramian~\cite{VandenHofEtAl2013}.

However, conventional persistence-of-excitation conditions do not directly translate to frequency-domain model estimation using complex \gls{gpr}. In this setting, the question is no longer if inputs are rich enough for a prescribed dynamic model, but rather whether several input-weighted \gls{gp}-supported latent spaces form a uniquely decomposable observation model.

Moreover, in addition to only considering \gls{siso} systems, existing \gls{gp}-based \gls{tf} estimators do not offer solutions that are appropriate for closed-loop systems where the addition of output disturbances can bias the estimates. Several approaches to estimating closed-loop systems are presented in the literature, e.g.,~\cite{ForssellLjung1999}, or for networked systems containing feedback~\cite{VandenHofEtAl2013}, but the bridge to \gls{gpr}-based methods is yet to be made.

This leads us to the contributions of this work. Specifically, we
\begin{enumerate}[label=(\arabic*)]
    \item extend existing open-loop \gls{siso}-\gls{gptf} estimators to the closed-loop \gls{miso} case. We first formulate the \gls{miso} estimator as the posterior mean of the distribution constructed by conditioning the \glspl{gp} representing the \glspl{tf} of interest, along with the transient, on measurements of the single output and multiple inputs;
    \item express exact module identifiability  as a direct-sum or rank condition on input-weighted prior factors, and show that the \gls{miso} estimator converges to the true modules in the noise-free limit;
    \item demonstrate how well-established  closed-loop bias-handling techniques can be employed in the \gls{gpr}-based \gls{tf} estimation setting; and
    \item verify the theoretical developments and efficacy of the developed estimator using \gls{mc} simulations of a simple dynamic power system model in both open- and closed-loop.
\end{enumerate}

The remainder of the paper is organized as follows: Section~\ref{sec:methods} sets up the problem at hand, extends existing \gls{siso}-\gls{gptf} estimators to the \gls{miso} case, develops the conditions required for exact module separability and identifiability, and considers how to handle systems with feedback. In Section~\ref{sec:results}, we present the power system case study, along with \gls{mc}-based verification of the developed identifiability criterion and presented closed-loop bias handling. Section~\ref{sec:discussion} contains a discussion, including considerations of the limitations of the work and future paths. Section~\ref{sec:conclusions} concludes the paper.
	\section{Methods}\label{sec:methods}
In this section, we present the central theoretical developments of this work. We formulate the \gls{miso} estimator based on complex \gls{gpr}, derive requirements for exact identifiability, and suggest how to handle closed-loop bias.

\subsection{Notation}
	For a generic single-argument scalar function \(\zeta\), we define its vectorized extension over a vector \(\bm{a} = \trans{[a_1, \dots, a_n]}\) as the element-wise mapping 
	\begin{equation}
		\zeta(\bm{a}) = \trans{
	\begin{bmatrix}
		\zeta(a_1) & \cdots & \zeta(a_n)
	\end{bmatrix}
	}.
	\end{equation}
	Similarly, for a two-argument scalar function \(\tilde{\zeta}\), evaluating the vectorized form \(\tilde{\zeta}(\bm{a}, \bm{b})\), where \(\bm{b} = \trans{[b_1, \dots, b_m]}\), corresponds to mapping the Cartesian product of the inputs to an \(n \times m\) matrix defined as
	\begin{equation}
	\tilde{\zeta}(\bm{a},\bm{b}) = 
    	\begin{bmatrix}
        	\tilde{\zeta}(a_1, b_1) & \dots & \tilde{\zeta}(a_1, b_m) \\
        	\vdots & \ddots & \vdots \\
        	\tilde{\zeta}(a_n, b_1) & \dots & \tilde{\zeta}(a_n, b_m)
   	 \end{bmatrix}.
	\end{equation}
		\subsection{Problem Formulation}\label{sec:problem_formulation}
We begin by considering the following generic time-domain \gls{miso} system
\begin{equation}\label{eq:time_domain_eq}
	y(t) = \sum_{i = 1}^m \left(u_i \ast g_i\right)(t) + v_y(t)
\end{equation}
where the output \(y(t)\) at time \(t\) is obtained by summing the convolutions of the individual impulse responses \(g_i(t)\) with their respective inputs \(u_i(t)\), as well as the additive noise \(v_y(t)\). \(m\) is the number of system modules. Going forward, we consider \(N\) samples of each input and the output, evaluated at \(t_n = nT_\mathrm{s}, \quad n = 0, \dots, N - 1\), associated with a measurement window of length \(T =  NT_\mathrm{s}\) and sampling period \(T_\mathrm{s}\). Taking the \(N\)-point \gls{dft} of the sampled and windowed input and output signals yields the frequency-domain input-output relation
\begin{equation}
	Y(\omega_k) = \sum_{i = 1}^m \left(U_i(\omega_k)G_i(\complex\omega_k) + T_i(\complex\omega_k)\right) + V_y(\omega_k).
\end{equation}
For real-valued signals, it is sufficient to consider only the nonnegative-frequency \gls{dft} bins \(k \in \{0, \dots, \allowbreak \lfloor N/2\rfloor\}\) with corresponding angular frequencies \(\omega_k = 2\pi k / (NT_\mathrm{s})\). The \(G_i(\complex\omega_k)\) are the continuous-time \glspl{tf} of the \(i\)th input-output channel, evaluated at the angular frequency \(\omega_k\), while \(U_i(\omega_k)\) and \(Y(\omega_k)\) are \gls{dft} coefficients corresponding to, respectively, the inputs and output. \(V_y(\omega_k)\) is the additive output noise in the frequency domain and the \(T_i(\complex\omega_k)\) are per-channel transient terms that arise from finite window effects caused by the mismatch between the periodic extension implicit in the \gls{dft} and the measured input-output trajectories~\cite{LataireChen2016}. We can express 
\begin{equation} 
	V_y(\omega_k) = V(\omega_k) + T_H(\complex\omega_k),
\end{equation} 
where \(T_H(\complex\omega_k)\) is a possible smooth noise transient, and introduce the following assumption.
\begin{assumption}[Output noise]\label{ass:output_noise}
	At the considered positive-frequency \gls{dft} bins, the stochastic additive output noise component is given by
	\begin{equation}
		V(\omega_k) = H(\complex\omega_k) E(\omega_k)
	\end{equation}
	where \(H(\complex\omega_k)\) is the frequency response of a stable noise-shaping filter. The sequence \(E(\omega_k)\) is zero-mean, jointly circular complex Gaussian, and uncorrelated across frequency bins with variance \(\sigma_E^2\), such that
	\begin{align}
		\expectation{E(\omega_k)} &= 0 \\
		\expectation{E(\omega_k)\overline{E(\omega_{k'})}} &= \sigma_E^2\delta_{kk'},~\text{and} \\
		\expectation{E(\omega_k)E(\omega_{k'})} &=0
	\end{align}
	where \(\delta_{kk'}\) is the Kronecker delta for \gls{dft} bins \(k\) and \(k'\),  and bar notation denotes complex conjugation.
\end{assumption}	
Under Assumption~\ref{ass:output_noise}, \(V(\omega_k)\) is zero-mean circular complex Gaussian and satisfies 
\begin{equation} 
	\expectation{V(\omega_k)\overline{V(\omega_{k'})}} = \sigma_V^2(\omega_k)\delta_{kk'}, 
\end{equation} 
with 
\begin{equation} 
	\sigma_V^2(\omega_k) = \expectation{\lvert V(\omega_k)\rvert^2} = \left\lvert H(\complex\omega_k)\right\rvert^2 \sigma_E^2.
\end{equation}
Circularity additionally implies 
\begin{equation} 
	\expectation{V(\omega_k)V(\omega_{k'})} = 0.
\end{equation}

As the module and noise transients all enter the output additively, they can be collected into a single equivalent transient. Since each contribution is a smooth function of angular frequency, the resulting equivalent transient is likewise modeled as a smooth, unknown function of angular frequency~\cite{HallemansEtAl2022}. By grouping \(T_H(\complex\omega_k)\) with the module transients as 
\begin{equation}
	T(\complex\omega_k) = \sum_{i=1}^m T_i(\complex\omega_k) + T_H(\complex\omega_k),
\end{equation}
we can express the \gls{miso} system as
\begin{equation}\label{eq:freq_domain_eq_transient}
	Y(\omega_k) = \sum_{i = 1}^m U_i(\omega_k)G_i(\complex\omega_k) + T(\complex\omega_k) + V(\omega_k).
\end{equation}

Equation \eqref{eq:freq_domain_eq_transient} reveals that, for estimation purposes, the equivalent transient can be represented as an additional artificial input channel with a unit-valued frequency-domain regressor. If we let \(G_0(\complex\omega_k) = T(\complex\omega_k)\) and \(U_0(\omega_k) = 1\) for all considered frequencies \(\omega_k\), the augmented frequency-domain \gls{miso} model becomes
\begin{equation}\label{eq:freq_domain_eq}
	Y(\omega_k) = \sum_{i = 0}^m U_i(\omega_k) G_i(\complex\omega_k) + V(\omega_k).
\end{equation}
Unless stated otherwise, any future references to modules or \glspl{tf} of interest include the transient. Now let \(\bm{k} = \trans{[k_1, \dots, k_F]}\) denote the ordered vector of \(F\) \gls{dft} bin indices in the frequency band of interest with corresponding angular frequencies \(\bm{\Omega_k} = \trans{\left[\omega_{k_1}, \dots, \omega_{k_F}\right]}\). The vectorized frequency-domain model can then be expressed as
\begin{equation}\label{eq:dft_model}
	\bm{Y} = \sum_{i = 0}^m \bm{D}_i\bm{G}_i + \bm{V}
\end{equation}
where \(\bm{G}_i = G_i(\complex\bm{\Omega_k})\), \(\bm{D}_i = \diag{U_i(\bm{\Omega_k})}\), \(\bm{Y} = Y(\bm{\Omega_k})\), and \(\bm{V} = V(\bm{\Omega_k})\).

Having established the frequency-domain system description, we are almost ready to move on to developing the \gls{miso} estimator. However, although both the input and output spectra are observed, they play different roles in the inference problem. The input spectra are assumed to be known exactly and are treated as fixed regressors, whereas the output is modeled as a random quantity through its dependence on the unknown module functions and stochastic noise. Accordingly, the probabilistic statements that follow should be understood as conditional on the observed input spectra, although, to improve readability, this conditioning is omitted from the notation.

For this conditional observation model to be valid, we make the following assumption.

\begin{assumption}[Regressor exogeneity]\label{ass:input_exogeneity} 
	Conditioning on the observed input spectra (that is, for \(i = 1, \dots, m\), excluding the virtual transient input) does not alter the stochastic output-noise model. In particular,
	\begin{align}
		\expectation{\bm{V} \mid \bm{U}_1, \dots, \bm{U}_m} &= \bm{0}~\text{and} \\
		\cov{\bm{V} \mid \bm{U}_1, \dots, \bm{U}_m} &= \cov{\bm{V}},
	\end{align}
	where \(\cov{\bm{V}}\) does not depend on the observed input realizations.
\end{assumption}

\subsection{MISO Estimator Based on Complex Gaussian Process Regression}
The \glspl{tf} we seek to identify are unknown infinite-dimensional objects, so we introduce a prior model in order to regularize the estimation problem and interpolate between frequency samples. Using complex \gls{gpr} turns the inference problem into Bayesian estimation of the \(G_i\), conditioned on the frequency-domain outputs.

In order to derive the posterior distribution, we require three covariance quantities, namely,
\begin{enumerate}[label=(\arabic*)]
	\item the prior covariance of the \glspl{tf} themselves;
    \item the covariance of the observed output spectrum; and
    \item the covariance between the output spectrum and the \glspl{tf} we want to estimate.    
\end{enumerate}
Once these quantities have been constructed, the posterior distribution follows directly from standard Gaussian conditioning formulas. Because the \glspl{tf}, inputs, and covariance functions are complex-valued, the derivations are first performed in the complex plane. Subsequently, the equivalent real-valued representation is introduced, allowing the use of standard multivariate Gaussian formulas and numerically stable computations.

\subsubsection{Covariance Matrices}
We first derive the \gls{gp} priors. Each module is modeled as a complex \gls{gp}, using the notation
\begin{equation}\label{eq:cgp}
	G_i(\complex\omega) \sim \mathcal{CGP}(0, K_i(\complex\omega, \complex\Omega), C_i(\complex\omega, \complex\Omega)),\quad \omega,\Omega\in\mathbb{R}.
\end{equation}
We assume zero-mean functions, so the complex \glspl{gp} are characterized by the covariance functions \(K_i(\complex\omega, \complex\Omega) = \expectation{G_i(\complex\omega)G_i(\complex\Omega)^\mathsf{H}}\)---where \(\mathsf{H}\) denotes conjugate transpose---relating the \gls{tf} values at angular frequencies \(\omega\) and \(\Omega\), and the complementary covariance functions \(C_i(\complex\omega, \complex\Omega) = \expectation{G_i(\complex\omega)G_i(\complex\Omega)^\top}\). Using the conjugate symmetry property of the Fourier transform gives
\begin{equation}\label{eq:comp_cov_func}
	C_i(\complex\omega, \complex\Omega) = K_i(\complex\omega, -\complex\Omega).
\end{equation}
Moreover, we also add the below assumption.
\begin{assumption}[Prior and noise independence] \label{ass:prior_independence} 
	Conditional on the observed regressors, the \glspl{gp} used to model the transfer functions \(G_i\), \(i = 0, \ldots, m\), together with the stochastic output-noise vector \(\bm{V}\), are mutually independent.
\end{assumption}

Next, we define the stacked module vector as
\begin{equation}
	\bm{G}(\complex\omega) = \trans{
	\begin{bmatrix}
		G_0(\complex\omega) & \cdots & G_m(\complex\omega)
	\end{bmatrix}
	}.
\end{equation}
Since the module priors are independent, the matrix-valued covariance and complementary covariance functions are
\begin{subequations}\label{eq:test_cov}
\begin{align}
	\bm{K_G}(\complex\omega, \complex\Omega) &= \expectation{\bm{G}(\complex\omega)\herm{\bm{G}(\complex\Omega)}}\notag \\
	&= \diag{K_0(\complex\omega, \complex\Omega), \dots, K_m(\complex\omega, \complex\Omega)}~\text{and} \\
	\bm{C_G}(\complex\omega, \complex\Omega) &= \expectation{\bm{G}(\complex\omega)\trans{\bm{G}(\complex\Omega)}}\notag \\
	&= \diag{C_0(\complex\omega, \complex\Omega), \dots, C_m(\complex\omega, \complex\Omega)},
\end{align}
\end{subequations}
meaning that we are not able to share information between modules. That is, we encode a relationship between the values of \(G_i(\complex\omega)\) and \(G_i(\complex\Omega)\), but not between \(G_i(\complex\omega)\) and \(G_j(\complex\Omega)\) for \(i \neq j\).

Using Assumptions~\ref{ass:output_noise}--\ref{ass:prior_independence}, together with the vectorized observation model \eqref{eq:dft_model}, the output covariance and complementary covariance matrices become, respectively,
\begin{subequations}\label{eq:train_cov}
\begin{align}
	\bm{K_Y}(\complex\bm{\Omega_k}, \complex\bm{\Omega_k}) &= \expectation{\bm{Y}\herm{\bm{Y}}}\notag \\ 
	&= \sum_{i = 0}^m \bm{D}_iK_i(\complex\bm{\Omega_k}, \complex\bm{\Omega_k})\herm{\bm{D}_i} + \diag{\sigma_V^2(\bm{\Omega_k})}~\text{and} \\
	\bm{C_Y}(\complex\bm{\Omega_k}, \complex\bm{\Omega_k}) &= \expectation{\bm{Y}\trans{\bm{Y}}}\notag \\
	&= \sum_{i = 0}^m \bm{D}_iC_i(\complex\bm{\Omega_k}, \complex\bm{\Omega_k})\trans{\bm{D}_i}.
\end{align}
\end{subequations}
\(\bm{Y} \in \mathbb{C}^{F\times 1}\) is distributed as a complex Gaussian vector with \(\bm{K_Y}\) and \(\bm{C_Y}\) both being members of \(\mathbb{C}^{F\times F}\).
\begin{remark}
	Since the \(\bm{D}_i\) are diagonal, their Hermitian transposes are obtained by conjugating the diagonal entries, such that \(\herm{\bm{D}_i} = \overline{\bm{D}_i}\). We nevertheless retain the notation \(\herm{\bm{D}_i}\) as it remains valid also in the case where \(\bm{D}_i\) is replaced by a more general linear operator.
\end{remark}

The final covariance matrix encodes the cross-covariance between outputs and modules. Assumption~\ref{ass:prior_independence} lets us express the matrix-valued cross-covariance and complementary cross-covariance functions as
\begin{subequations}\label{eq:cross_cov}
\begin{align}
	\bm{K_{GY}}(\complex\omega, \complex\bm{\Omega_k}) &= \expectation{\bm{G}(\complex\omega)\herm{\bm{Y}}}\notag \\
	&= \trans{
	\begin{bmatrix}
		K_0(\complex\omega, \complex\bm{\Omega_k})\herm{\bm{D}_0} & \cdots & K_m(\complex\omega, \complex\bm{\Omega_k})\herm{\bm{D}_m}
	\end{bmatrix}
	}~\text{and} \\	
    \bm{C_{GY}}(\complex\omega, \complex\bm{\Omega_k}) &= \expectation{\bm{G}(\complex\omega)\trans{\bm{Y}}}\notag \\
	&= \trans{
	\begin{bmatrix}
		C_0(\complex\omega, \complex\bm{\Omega_k})\trans{\bm{D}_0} & \cdots & C_m(\complex\omega, \complex\bm{\Omega_k})\trans{\bm{D}_m}
	\end{bmatrix}
	}.
\end{align}
\end{subequations}

\subsubsection{Real-Expanded Representation}
In order to use standard multivariate Gaussian conditioning formulas, we represent each complex vector \(\bm{z} \in \mathbb{C}^a\) as its corresponding real composite vector, denoting it with the check accent as
\begin{equation}
	\check{\bm{z}} = \trans{
	\begin{bmatrix}
		\re{\bm{z}} & \im{\bm{z}}
	\end{bmatrix}
	} \in \mathbb{R}^{2a}.
\end{equation}
Similarly, for a complex matrix \(\bm{A} \in \mathbb{C}^{b\times a}\), we let the corresponding real representation be
\begin{equation}
	\check{\bm{A}} = 
	\begin{bmatrix} 
		\re{\bm{A}} & -\im{\bm{A}} \\ 
		\im{\bm{A}} & \re{\bm{A}} 
	\end{bmatrix} 
	\in\mathbb R^{2b\times2a}.
\end{equation}
This way the vector and matrix representations have the multiplicative property
\begin{equation}
	\check{(\bm{A}\bm{z})} = \check{\bm{A}}\check{\bm{z}}.
\end{equation}
Following derivations from~\cite{Lapidoth2017}, for a general covariance-complementary covariance pair \((\bm{K}',\bm{C}')\), we define the corresponding real covariance matrix using the operator \(\bm{\Gamma}(\cdot)\) as
\begin{equation}\label{eq:real_expand}
	\bm{\Gamma}(\bm{K}',\bm{C}')
	= \frac{1}{2}
	\begin{bmatrix}
		\re{\bm{K}' + \bm{C}'} & \im{\bm{C}' - \bm{K}'} \\
		\im{\bm{C}' + \bm{K}'} & \re{\bm{K}' - \bm{C}'}
	\end{bmatrix}.
\end{equation}
Applying the transformation \eqref{eq:real_expand} to \eqref{eq:test_cov}--\eqref{eq:cross_cov} gives the real-valued constituent covariances
\begin{subequations}\label{eq:covs}
\begin{align}
	\bm{K}(\complex\bm{\Omega_k}, \complex\bm{\Omega_k}) &= \bm{\Gamma}(\bm{K_Y}, \bm{C_Y}), \\
	\bm{K}_\ast(\complex\omega, \complex\bm{\Omega_k}) &= \bm{\Gamma}(\bm{K_{GY}}, \bm{C_{GY}})~\text{and} \\
	\bm{K}_{\ast\ast}(\complex\omega, \complex\Omega) &= \bm{\Gamma}(\bm{K_G}, \bm{C_G}),
\end{align}
\end{subequations}
where we have dropped the inner arguments to ease notation. 

By \eqref{eq:cgp}, Assumptions~\ref{ass:output_noise}--\ref{ass:prior_independence}, and the linear observation model \eqref{eq:dft_model}, the real-expanded module functions \(\check{\bm{G}}\) and output measurements \(\check{\bm{Y}}\) are jointly Gaussian for every finite collection of angular frequencies with distribution
\begin{equation}\label{eq:joint_dist}
	\begin{bmatrix}
		\check{\bm{G}}(\complex\omega) \\
		\check{\bm{Y}}
	\end{bmatrix}
	\sim \mathcal{N} 
	\left(
	\begin{bmatrix}
		\bm{0} \\
		\bm{0}
	\end{bmatrix},
	\begin{bmatrix}
		\bm{K}_{\ast\ast}(\complex\omega, \complex\Omega) & \bm{K}_\ast(\complex\omega, \complex\bm{\Omega_k}) \\  
		\trans{\bm{K}_\ast(\complex\Omega, \complex\bm{\Omega_k})} & \bm{K}(\complex\bm{\Omega_k}, \complex\bm{\Omega_k})
	\end{bmatrix}
	\right).
\end{equation}
From this we can construct their conditional distribution and formulate the \gls{miso-gptf} estimator.
\begin{proposition}[\gls{miso-gptf} estimator]
	Conditioning the module functions on the output measurements yields the posterior \gls{gp}
	\begin{equation}
		\check{\bm{G}}(\complex\omega) \mid \check{\bm{Y}} \sim \mathcal{GP}(\check{\bm{\mu}}(\complex\omega), \bm{\Sigma}(\complex\omega, \complex\Omega)),\quad \omega,\Omega\in\mathbb{R},
	\end{equation}
	where
	\begin{align}
		\check{\bm{\mu}}(\complex\omega) &= \bm{K}_\ast(\complex\omega, \complex\bm{\Omega_k})\inv{\bm{K}(\complex\bm{\Omega_k}, \complex\bm{\Omega_k})}\check{\bm{Y}}~\text{and} \\
		\bm{\Sigma}(\complex\omega, \complex\Omega) &= \bm{K}_{\ast\ast}(\complex\omega, \complex\Omega) - \bm{K}_\ast(\complex\omega, \complex\bm{\Omega_k})\inv{\bm{K}(\complex\bm{\Omega_k}, \complex\bm{\Omega_k})}\trans{\bm{K}_\ast(\complex\Omega, \complex\bm{\Omega_k})}.
	\end{align}
	For each angular frequency \(\omega\), the pointwise \gls{map} estimate coincides with the posterior mean, and we construct the complex \gls{miso-gptf} estimates from
	\begin{equation} 
		\hat{\bm{G}}(\complex\omega) = 
		\begin{bmatrix}
			\bm{I}_{m+1} & \complex\bm{I}_{m+1}
		\end{bmatrix}
		\check{\bm{\mu}}(\complex\omega)
	\end{equation}
	where \(\bm{I}_{m+1}\) is the \((m + 1)\times(m + 1)\) identity matrix. Hat notation is used to denote an estimate.
\end{proposition}
\begin{proof}
	The posterior mean and covariance follow directly from applying the standard Gaussian conditioning formulas to \eqref{eq:joint_dist}. Since the conditional marginal distribution at each \(\omega\) is Gaussian, its mean coincides with its mode and is therefore the pointwise \gls{map} estimate.
\end{proof}
\begin{remark}
At \(\Omega=\omega\), the entries of \(\bm\Sigma\) associated with the real and imaginary components of an individual module function give its marginal posterior covariance. The remaining entries describe posterior cross-covariances between the module functions, which may be nonzero because they are conditioned on the same output measurements.
\end{remark}
\begin{remark}
	By letting \(m = 1\), it is straightforward to see that the \gls{miso-gptf} estimator is a generalization of existing estimators based on complex \gls{gpr} for \gls{siso} systems, like~\cite{HallemansEtAl2022, LataireChen2016}.
\end{remark}

Selection of appropriate covariance functions and optimizing their hyperparameters is beyond the scope of this paper, but the interested reader is referred to~\cite{NarayanEtAl2026}.

Having introduced the \gls{miso-gptf} estimator, the next step, and the main contribution of this paper, is to determine when it is possible to separately identify the modules from recordings of the single noisy output and multiple inputs.
		\subsection{Identifiability of Individual Modules}\label{sec:module_identifiability}
In this subsection, we derive conditions under which the \gls{miso-gptf} estimator will be able to identify the individual modules of the \gls{miso} system. We begin with a structural identifiability result before moving to the \gls{gpr} setting.

\subsubsection{Structural Module Identifiability}
Let \(\mathcal{H}_i\) denote the function space in which \(G_i\) is to be identified. In a \gls{gp} setting, \(\mathcal{H}_i\) may be chosen as a function space induced by the covariance function, or, on a finite frequency grid, as the prior-supported subspace determined by the covariance matrix. Define the multiplication operator
\begin{equation}\label{eq:multiplication-operator}
	\mathcal{M}_i \colon \mathcal{H}_i \to \mathcal{Y}, \quad (\mathcal{M}_iG_i)(\omega) = U_i(\omega)G_i(\complex\omega),
\end{equation}
where \(\mathcal{Y}\) is the observation space. \(\mathcal{M}_i\) maps the unknown transfer function \(G_i\) from the space \(\mathcal{H}_i\) to the contribution it would produce in the observation space \(\mathcal{Y}\). We assume the following throughout:
\begin{assumption}[Well-defined multiplication operator]
	Each multiplication operator \(\mathcal{M}_i\) is well defined, i.e., \(U_i(\omega)G_i(\complex\omega) \in \mathcal{Y}\) for all \(G_i \in \mathcal{H}_i\), and each \(\mathcal{M}_i\) is bounded.
\end{assumption}

Next, define the product space
\begin{equation}
	\mathcal{H} = \mathcal{H}_0 \oplus \dots \oplus \mathcal{H}_m = \left\{(G_0, \dots, G_m) : G_i \in \mathcal{H}_i\right\},
\end{equation}
i.e., \(\mathcal{H}\) is defined as the external direct sum of the (finitely many) \(\mathcal{H}_i\), denoted by \(\oplus\). \(\mathcal{H}\) is the space of ordered tuples of the \(G_i\). Define the observation operator
\begin{equation}\label{eq:observation-operator}
	\mathcal{A} \colon \mathcal{H} \to \mathcal{Y}, \quad \mathcal{A}(G_0, \dots, G_m) = \sum_{i = 0}^m \mathcal{M}_i G_i.
\end{equation}
This is the forward model and answers the question \say{if the true transfer functions were exactly \(G_0,\dots, G_m\), what would the total noise-free observation be?} \(\mathcal{A}\) multiplies each \(G_i\) by its known input \(U_i\) (via \(\mathcal{M}_i\)) and sums their contributions. As such, \(\mathcal{A}\) maps the set of unknowns to a single predicted output in \(\mathcal{Y}\). Since each \(\mathcal{M}_i\) is linear, \(\mathcal{A}\) is itself linear.

We then define the output contributions as
\begin{equation}
	Z_i = \mathcal{M}_i G_i.
\end{equation}
Each \(Z_i\) is an element in the observation space \(\mathcal{Y}\), representing what module \(i\) alone contributes to the total output. The associated input-weighted module spaces are defined as
\begin{equation}\label{eq:input-weighted-module-space}
	\mathcal{S}_i = \mathcal{M}_i \mathcal{H}_i = \left\{\mathcal{M}_iG_i : G_i \in \mathcal{H}_i\right\} \subseteq \mathcal{Y}.
\end{equation}
This is the range of the operator \(\mathcal{M}_i\). It is the subspace of \(\mathcal{Y}\) spanned by all possible contributions from module \(i\), given the input \(U_i\).

We are now ready to state a purely structural identifiability result.
\begin{definition}[Exact module identifiability]
	The modules \(G_0,\dots, G_m\) are \emph{exactly identifiable} on \(\mathcal{H}\) if 
	\begin{equation}
		\mathcal{A}(G_0,\dots, G_m) = \mathcal{A}(\tilde{G}_0,\dots, \tilde{G}_m)
	\end{equation}
	implies that 
	\begin{equation}
		G_i = \tilde{G}_i,\quad i = 0,\dots, m.
	\end{equation}
\end{definition}
This is the same as requiring that \(\mathcal{A}\) is injective. We can now formulate the following theorem.
\begin{theorem}[Exact identifiability via direct sums]\label{theo:exact_identifiability}
	Let
	\begin{equation}
		\mathcal{S}_i = \mathcal{M}_i \mathcal{H}_i \subseteq \mathcal{Y}
	\end{equation}
	denote the space of possible output contributions from module \(i\). Consider the additive observation model
	\begin{equation}
		Y = \sum_{i = 0}^m Z_i, \quad Z_i = \mathcal{M}_iG_i.
	\end{equation}
	The output contributions \(Z_i\) are uniquely identifiable if and only if
	\begin{equation}
		\mathcal{S}_0 + \dots + \mathcal{S}_m
	\end{equation}
	is a direct sum, i.e.,
	\begin{equation}\label{eq:direct-sum-zero-sum}
		s_0 + \dots + s_m = 0,\quad s_i \in \mathcal{S}_i
	\end{equation}
	implies
	\begin{equation}\label{eq:direct-sum-zero-elements}
		s_0 = \dots = s_m = 0.
	\end{equation}
	The modules \(G_i\) are themselves exactly identifiable if and only if, in addition, each \(\mathcal{M}_i\) is injective on \(\mathcal{H}_i\). Equivalently, the combined observation operator \(\mathcal{A}\) is injective.
\end{theorem}
In short, recovering each module requires two conditions: (1) the spaces of module contributions must not overlap; and (2) each map from module to output must be injective on its admissible space. This makes it possible to uniquely attain \(G_i\) from  \(Z_i\) when \(U_i\) is known.

\begin{proof}
	The proof consists of showing that the contributions forming a direct sum and being unique are equivalent, and that injectivity of \(\mathcal{A}\) implies both that each \(\mathcal{M}_i\) is injective and that the module spaces \(\mathcal{S}_i\) form a direct sum. The full proof is given in \ref{apx:exact_identifiability}.
\end{proof}

\subsubsection{Identifiability in a MISO GPR Setting}
Up until now, we have considered identifiability from a deterministic perspective. We now extend this perspective to the finite-dimensional Gaussian random vectors induced by evaluating the module \glspl{gp} on the selected \gls{dft} grid. 

First let us introduce some prerequisite quantities before introducing a corollary of Theorem~\ref{theo:exact_identifiability}. Consider the real-expanded version of the vectorized fre\-quen\-cy-domain model in \eqref{eq:dft_model},
\begin{equation}\label{eq:real_dft_model}
	\check{\bm{Y}} = \sum_{i = 0}^m \check{\bm{D}}_i\check{\bm{G}}_i + \check{\bm{V}}
\end{equation}
along with Assumption~\ref{ass:prior_independence} and \eqref{eq:cgp}. The grid-evaluated module vectors are then mutually independent and satisfy 
\begin{equation} 
	\check{\bm{G}}_i \sim \mathcal N(\bm{0}, \bm{\Gamma}_i), \qquad i=0,\ldots, m, 
\end{equation}
where 
\begin{equation}
	\bm{\Gamma}_i = \bm{\Gamma}(\bm{K}_i, \bm{C}_i)
\end{equation}
is obtained from the covariance and complementary covariance matrices of the \(i\)th module on the selected frequency grid. Let
\begin{equation}
	r_i = \rank{\bm{\Gamma}_i}
\end{equation}
and let \(\bm{\Gamma}_i^{1/2} \in \mathbb{R}^{2F\times r_i}\) denote a full-column-rank covariance factor satisfying 
\begin{equation}
	\bm{\Gamma}_i^{1/2}\trans*{\bm{\Gamma}_i^{1/2}} = \bm{\Gamma}_i.
\end{equation}		
For each module, the prior-supported and corresponding output contribution subspaces are
\begin{align}
	\mathcal{H}_i &= \range{\bm{\Gamma}_i^{1/2}}~\text{and} \\
	\mathcal{T}_i &= \check{\bm{D}}_i\mathcal{H}_i = \range{\check{\bm{D}}_i\bm{\Gamma}_i^{1/2}} \subseteq \mathbb{R}^{2F}.
\end{align}
Finally, define the concatenated matrix
\begin{equation}\label{eq:contribution_matrix}
	\bm{L} = 
	\begin{bmatrix}
		\check{\bm{D}}_0\bm{\Gamma}_0^{1/2} & \cdots & \check{\bm{D}}_m\bm{\Gamma}_m^{1/2}
	\end{bmatrix}.
\end{equation}

\begin{corollary}[Exact identifiability of \gls{gp}-supported \gls{miso} modules on a finite angular-frequency grid]\label{cor:posterior_consistency}
	Consider the real-expanded frequency-domain model \eqref{eq:real_dft_model}. The GP-supported output contributions \(\check{\bm{Z}}_i = \check{\bm{D}}_i\check{\bm{G}}_i\) are exactly identifiable in the noiseless limit if and only if
	\begin{equation}
		\mathcal{T}_0 + \dots + \mathcal{T}_m
	\end{equation}
	is a direct sum. Equivalently, the contributions are exactly identifiable if and only if
	\begin{equation}
		\rank{\bm{L}} = \sum_{i = 0}^m \rank{\check{\bm{D}}_i\bm{\Gamma}_i^{1/2}}.
	\end{equation}
	The grid-evaluated module vectors are themselves exactly identifiable on their GP-supported subspaces if and only if the contributions are exactly identifiable and, additionally,
	\begin{equation}
		\kernel{\check{\bm{D}}_i} \cap \mathcal{H}_i = \{\bm{0}\},\quad i = 0, \dots, m.
	\end{equation}
	Module-level injectivity is equivalent to
	\begin{equation}
		\rank{\check{\bm{D}}_i\bm{\Gamma}_i^{1/2}} = \rank{\bm{\Gamma}_i^{1/2}} = r_i,\quad i = 0,\ldots, m.
	\end{equation}
	Combining the separability and injectivity requirements gives that modules are exactly identifiable on their \gls{gp}-supported subspaces in the noiseless limit if and only if
	\begin{equation}\label{eq:rank_condition}
		\rank{\bm{L}} = \sum_{i = 0}^m r_i.
	\end{equation}
\end{corollary}
\begin{proof}
	The result follows directly from applying Theorem~\ref{theo:exact_identifiability} to the finite-di\-men\-sion\-al prior-supported subspaces \(\mathcal{H}_i = \range{\bm{\Gamma}_i^{1/2}}\) and the restricted real-expanded multiplication maps
	\begin{equation}
		\mathcal{M}_i \colon \mathcal{H}_i \to \mathbb{R}^{2F}, \quad \mathcal{M}_i(\bm{\gamma}) = \check{\bm{D}}_i\bm{\gamma},
	\end{equation}
	where \(\bm{\gamma}\in \mathcal{H}_i\) denotes an arbitrary prior-supported module vector. Since \(\mathcal{M}_i\mathcal {H}_i=\mathcal{T}_i\), the direct-sum and injectivity conditions in Theorem~\ref{theo:exact_identifiability} yield the stated conditions.
\end{proof}

\begin{remark}\label{rem:effective_rep}
	In order to fulfill the requirement in \eqref{eq:rank_condition}, all directions supported by the covariance matrices \(\bm{\Gamma}_i\) have to be exactly identifiable on the finite frequency grid. It is common for smooth covariance functions to produce full-rank covariance matrices when evaluated on a finite grid. However, the resulting eigenvalues may decay rapidly, meaning that exact identifiability may fail despite the \gls{gp} posterior still managing to separate the dominant components in practice. To address this, we can replace \(\bm{\Gamma}_i^{1/2}\) by an effective covariance factor. We first compute the eigendecomposition \(\bm{\Gamma}_i = \bm{Q}_i\bm{\Lambda}_i\trans{\bm{Q}_i}\) where \(\bm{\Lambda}_i = \diag{\lambda_{i,1},\dots, \lambda_{i,2F}}\) with \(\lambda_{i,1} \geq \lambda_{i,2} \geq \dots \geq \lambda_{i,2F} \geq 0\). Next, we retain only eigenvalues greater than \(\varepsilon\lambda_{i,1}\), where \(0 < \varepsilon < 1\) sets the threshold, and replace \(\bm{\Gamma}_i^{1/2}\) with the effective covariance factor \(\bm{\Gamma}_{i,\varepsilon}^{1/2} = \bm{Q}_{i,\varepsilon}\bm{\Lambda}_{i,\varepsilon}^{1/2}\). 
\end{remark}
Remark~\ref{rem:effective_rep} leads to the following corollary.

\begin{corollary}[Effective \gls{gp}-supported identifiability]\label{cor:effective_identifiability}
	Let
	\begin{equation}
		\bm{\Gamma}_i = \bm{Q}_i\bm{\Lambda}_i\trans{\bm{Q}_i}
	\end{equation}
	be the eigendecomposition of the real-expanded finite-grid covariance matrix of module \(i\). For tolerance \(0 < \varepsilon < 1\), let 
	\begin{equation}
		\bm{\Gamma}_{i,\varepsilon}^{1/2} = \bm{Q}_{i,\varepsilon}\bm{\Lambda}_{i,\varepsilon}^{1/2}
	\end{equation}
	denote the covariance factor obtained by retaining only eigenvalues larger than \(\varepsilon\lambda_{i,1}\). Define the effective prior-supported subspace as
	\begin{equation}
		\mathcal{H}_{i,\varepsilon} = \range{\bm{\Gamma}_{i,\varepsilon}^{1/2}}, 
	\end{equation}
	and define the effective concatenated contribution matrix as
	\begin{equation} 
		\bm{L}_\varepsilon =
		\begin{bmatrix}
			\check{\bm{D}}_0 \bm{\Gamma}_{0,\varepsilon}^{1/2} & \cdots & \check{\bm{D}}_m \bm{\Gamma}_{m,\varepsilon}^{1/2}
		\end{bmatrix}.
	\end{equation}
	Then the modules are exactly identifiable on their effective prior-supported subspaces if and only if
	\begin{equation}
		\rank{\bm{L}_\varepsilon} = \sum_{i = 0}^m \rank{\bm{\Gamma}_{i,\varepsilon}^{1/2}}.
	\end{equation}
\end{corollary}
\begin{proof}
	The result follows from applying Theorem~\ref{theo:exact_identifiability} to the effective prior-sup\-por\-ted subspaces \(\mathcal{H}_{i,\varepsilon} = \range{\bm{\Gamma}_{i,\varepsilon}^{1/2}}\) and the restricted real-expanded multiplication maps
	\begin{equation}
		\mathcal{M}_{i,\varepsilon} \colon \mathcal{H}_{i,\varepsilon} \to \mathbb{R}^{2F}, \quad \mathcal{M}_{i,\varepsilon}(\bm{\gamma}) = \check{\bm{D}}_i\bm{\gamma},
	\end{equation}
	where \(\bm{\gamma}\in \mathcal{H}_{i,\varepsilon}\) is an arbitrary effective prior-supported module vector.
\end{proof}

The effective condition does not imply strict exact identifiability of the full \gls{gp} model. It states that the data can uniquely separate the retained dominant prior directions. The discarded directions are treated as negligible at the chosen tolerance.

Corollaries~\ref{cor:posterior_consistency} and~\ref{cor:effective_identifiability} give structural conditions for exact identifiability on the finite frequency grid. The next question is then whether these conditions are reflected in the actual \gls{miso-gptf} estimator.

\subsubsection{Zero-Noise Consistency of the Posterior Mean}
We now examine if the posterior mean, i.e., the \gls{miso-gptf} estimate, recovers the true grid-evaluated \gls{tf} vectors as the additive noise variance tends to zero. In the following analysis, \(\bm{\Gamma}_i^{1/2}\) denotes the chosen covariance factor, be that the full or effective version, and \(\bm{L}\) denotes the stacked matrix associated with the chosen prior representation. The derivations are identical in both cases, only the interpretation of the prior-supported subspaces changes.

For each module, write its real-expanded grid-evaluated vector in terms of the covariance factor \(\bm{\Gamma}_i^{1/2}\) as
\begin{equation}\label{eq:module_whitened_representation} 
	\check{\bm{G}}_i = \bm{\Gamma}_i^{1/2}\bm{W}_i, \quad \bm{W}_i \sim \mathcal N(\bm{0},\bm{I}_{r_i}), 
\end{equation}
i.e., its whitened basis. Since the module priors are mutually independent, the stacked latent vector
\begin{equation} 
	\bm{W}= \trans{
	\begin{bmatrix}
		\trans{\bm{W}}_0 & \cdots & \trans{\bm{W}}_m
	\end{bmatrix}
	} \in \mathbb{R}^R,
\end{equation}
where
\begin{equation}
	R = \sum_{i=0}^{m}r_i, 
\end{equation}
satisfies 
\begin{equation} 
	\bm{W} \sim \mathcal{N}(\bm{0},\bm{I}_R). 
\end{equation}
Substituting \eqref{eq:module_whitened_representation} into the real-expanded observation model \eqref{eq:real_dft_model} gives 
\begin{equation}\label{eq:latent_observation_model} 
	\check{\bm{Y}} = \bm{L}\bm{W} + \check{\bm{V}}, 
\end{equation} 
where \(\bm{L} \in \mathbb R^{2F\times R}\) is as defined in \eqref{eq:contribution_matrix}. The stochastic noise model introduced in Assumption~\ref{ass:output_noise} generally gives a fre\-quen\-cy-dependent noise covariance, so we write the covariance of \(\check{\bm{V}}\) as \(\sigma^2\bm{R_V}\) where \(\sigma^2\) controls the overall noise level and \(\bm{R_V}\) is a fixed positive-definite noise-shaping matrix. Now let \(\bm{P_V}\) be a whitening matrix satisfying
\begin{equation}
	\bm{P_VR_V}\trans{\bm{P_V}} = \bm{I}_{2F}.
\end{equation}
Defining 
\begin{subequations}
\begin{align} 
	\bm{Y}^{(W)} &= \bm{P_V}\check{\bm{Y}}, \\\
	\bm{L}^{(W)} &= \bm{P_V}\bm{L},\label{eq:prewhitened_L}~\text{and} \\
	\bm{\epsilon} &= \inv{\sigma}\bm{P_V}\check{\bm{V}}, 
\end{align} 
\end{subequations}
gives the prewhitened model 
\begin{equation}\label{eq:prewhitened_model} 
	\bm{Y}^{(W)} = \bm{L}^{(W)}\bm{W} + \sigma\bm{\epsilon}, \quad \bm{\epsilon} \sim \mathcal{N}(\bm{0},\bm{I}_{2F}). 
\end{equation}
As \(\bm{P_V}\) is nonsingular, prewhitening does not change the rank or exact identifiability of the forward model. Note how Assumption~\ref{ass:prior_independence} also implies that \(\bm{W}\) is independent of both \(\check{\bm{V}}\) and its prewhitened representation \(\bm{\epsilon}\). Applying the standard conditional mean formula for Gaussian processes to \eqref{eq:prewhitened_model} therefore gives the conditional mean of \(\bm{W}\) conditioned on \(\bm{Y}\) as
\begin{equation}\label{eq:post_latent_mean}
	\hat{\bm{W}} = \trans*{\bm{L}^{(W)}}\inv*{\bm{L}^{(W)}\trans*{\bm{L}^{(W)}} + \sigma^2\bm{I}_{2F}}\bm{Y}^{(W)}.
\end{equation}
Partitioning \(\hat{\bm{W}}\) into module-specific blocks \(\hat{\bm{W}}_i\), together with \eqref{eq:module_whitened_representation} and the linearity of conditional expectation, gives the real-expanded \gls{miso-gptf} estimate as the covariance factor acting on the corresponding block of the latent conditional mean, i.e.,
\begin{equation}\label{eq:mean_whitened_representation}
	\hat{\check{\bm{G}}}_i = \bm{\Gamma}_i^{1/2}\hat{\bm{W}}_i.
\end{equation}

In order to characterize the small noise behavior of the estimate, consider the compact \gls{svd}
\begin{subequations}\label{eq:svd}
\begin{align}
	\bm{L}^{(W)} &= \bm{U}_\mathrm{s}\bm{S}\trans{\bm{V}_\mathrm{s}},~\text{where} \\
	\bm{S} &= \diag{s_1, \dots, s_P},~\text{and} \\
	P &= \rank{\bm{L}^{(W)}}.
\end{align}
\end{subequations}
Inserting \eqref{eq:prewhitened_model} and \eqref{eq:svd} into \eqref{eq:post_latent_mean} lets us express
\begin{equation}\label{eq:latent_mean_svd} 
	\hat{\bm{W}} = \bm{V}_\mathrm{s} \diag{\frac{s_p^2}{s_p^2 + \sigma^2}} \trans{\bm{V}_\mathrm{s}} \bm{W} + \sigma \bm{V}_\mathrm{s} \diag{ \frac{s_p}{s_p^2 + \sigma^2}} \trans{\bm{U}_\mathrm{s}} \bm{\epsilon},
\end{equation}
where \(s_p\) is the \(p\)th singular value. The first term of \eqref{eq:latent_mean_svd} describes the shrinkage of the true latent coefficients, and the second describes the contribution of measurement noise to the posterior mean. The full derivation of \eqref{eq:latent_mean_svd} is given in \ref{apx:miso-gptf_consistency}.

With this established, we can introduce the following proposition.
\begin{proposition}[Zero-noise consistency of the \gls{miso-gptf} estimator]\label{prop:miso-gptf_consistency}
	Consider the finite-grid \gls{gp} model of Corollary~\ref{cor:posterior_consistency} and suppose that the module-identifiability condition
	\begin{equation}\label{eq:rank_condition_holds}
		\rank{\bm{L}} = \sum_{i = 0}^m r_i = R
	\end{equation}
	holds. Assume that the real-expanded noise covariance can be written as \(\sigma^2\bm{R_V}\), where \(\bm{R_V}\) is fixed and positive definite. Then, as \(\sigma^2 \to 0\),
	\begin{equation}
		\hat{\bm{W}} \to \bm{W}
	\end{equation}
	in \(L^2\) and hence in probability. Consequently, for every module \(i = 0, \dots, m\),
	\begin{equation}
		\hat{\check{\bm{G}}}_i \to \check{\bm{G}}_i
	\end{equation}
	in \(L^2\) and probability.
\end{proposition}
\begin{proof}
	After prewhitening, the model has the form \eqref{eq:prewhitened_model}. Since the whitening matrix \(\bm{P_V}\) is nonsingular, \eqref{eq:prewhitened_L} and \eqref{eq:rank_condition_holds} imply that \(\bm{L}^{(W)}\) has full column rank. The \gls{svd} representation \eqref{eq:svd} then shows that every shrinkage factor converges to one and that every noise term converges to zero in \(L^2\) as \(\sigma^2 \to 0\). Hence \(\hat{\bm{W}} \to \bm{W}\) and since \(\bm{\Gamma}_i^{1/2}\) is a fixed finite-dimensional linear operator, the corresponding module posterior means satisfy \(\hat{\check{\bm{G}}}_i \to \check{\bm{G}}_i\). A complete proof is given in \ref{apx:miso-gptf_consistency}.
\end{proof}

\begin{remark}
	If the forward model is rank-deficient, that is if
	\begin{equation}
		P = \rank{\bm{L}} < R,
	\end{equation}
	then
	\begin{equation}
		\hat{\bm{W}} \to \bm{V}_\mathrm{s}\trans{\bm{V}_\mathrm{s}}\bm{W}
	\end{equation}
	in \(L^2\) as \(\sigma^2 \to 0\). The matrix \(\bm{V}_\mathrm{s}\trans{\bm{V}_\mathrm{s}}\) projects \(\bm{W}\) onto the subspace spanned by the columns of \(\bm{V}_\mathrm{s}\). The columns of \(\bm{V}_\mathrm{s}\) are exactly the right singular vectors of \(\bm{L}^{(W)}\) with nonzero singular values, that is the row space of \(\bm{L}^{(W)}\). Therefore, the estimator only recovers the components of \(\bm{W}\) observable through the forward model, i.e., those not in the nullspace of \(\bm{L}^{(W)}\), and the posterior covariance hence does not collapse along these unobservable directions.
\end{remark}

\begin{remark}
	If the effective representation is used, that is \(\bm{\Gamma}_{i,\varepsilon}^{1/2}\) and \(\bm{L}_\varepsilon\), Proposition~\ref{prop:miso-gptf_consistency} applies to the retained effective components \(\hat{\check{\bm{G}}}_{i,\varepsilon} = \bm{\Gamma}_{i,\varepsilon}^{1/2}\hat{\bm{W}}_{i,\varepsilon}\), and so does not imply recovery of the discarded low-variance directions of the full \gls{gp} prior.
\end{remark}
		\subsection{Closed-Loop Identifiability}\label{sec:closed-loop_identifiability}
Up until this point, we have only focused on the separability and identifiability issue, i.e., if the collected data and setup are able to separate each \gls{tf} in a \gls{miso} system. However, the preceding derivations rely on Assumption~\ref{ass:input_exogeneity}, and so, as a prior step, the structural issue of obtaining valid inputs/regressors for the problem we intend to solve has to be resolved. In other words, we require inputs that excite the dynamics we intend to map, while also ensuring that they are uncorrelated with the added noise.

This issue becomes particularly pertinent when we try to apply the \gls{miso-gptf} estimator to a closed-loop system. The preceding open-loop conditional likelihood is no longer valid as part of the output disturbance is statistically predictable from the measured regressors and may incorrectly be attributed to the module dynamics. This could potentially bias the result.

Thankfully, well-established techniques exist for handling closed-loop bias \cite{ForssellLjung1999}, any number of which are compatible with the \gls{miso-gptf} estimator. For the sake of completeness, we here briefly describe a basic method of obtaining usable regressors relying on least-squares projection. The effects of using structural, i.e., actual module inputs \(u_i\), compared to projected inputs, in the context of the \gls{miso-gptf} estimator, are investigated further in Section~\ref{sec:results}.

If we again consider the time-domain description in \eqref{eq:time_domain_eq}, we can close the loops by introducing the feedback equations
\begin{equation}\label{eq:time_domain_cl}
	u_i(t) = \left(c_i \ast \left(r_i - y_i^{(\mathrm{fb})}\right)\right)(t) \quad \text{for}~i = 1, \dots, m,
\end{equation}
where \(r_i\) are measured, exogenous signals, typically controller references or excitations, and \(c_i\) are feedback controllers. \(y_i^{(\mathrm{fb})}\) is the output signal available to controller \(i\)---be that the \gls{miso} output \(y\) or some other local signals in the general case. From \eqref{eq:time_domain_eq} and \eqref{eq:time_domain_cl} it is evident that the measured inputs generally contain both components generated by the excitation \(r_i\) and by the disturbance \(v_y\). For input \(i = 1, \dots, m\), the decomposition can be denoted as
\begin{equation}\label{eq:input_decomp}
	u_i(t) = u_i^{(r)}(t) + u_i^{(v)}(t),
\end{equation}
where \(u_i^{(r)}\) and \(u_i^{(v)}\) are the components induced by the excitation and disturbance, respectively. We construct regressors that contain predominantly excitation-generated input components via Algorithm~\ref{algo:construct_regressors}.
\begin{algorithm}[tb] 
\caption{Construction of regressors from endogenous inputs}\label{algo:construct_regressors} 
\begin{algorithmic}[1] 
    \Require Structural input set \(\mathcal{U}\) and measured external excitation signals \(r_i\) 
    \Ensure Excitation set \(\mathcal{R}\), relevant input subset \(\mathcal{U}_{\mathcal{R}}\), and projected input set \(\mathcal{U}_{\mathcal{R}}^{(r)}\) 
    
    \State Choose \(\mathcal{R} \gets \{r_1,\ldots,r_q\}\) such that each \(r_j\) excites at least one \(u_i \in \mathcal{U}\) 
    \State Set \(\mathcal{U}_{\mathcal{R}} \gets \{u_i \in \mathcal{U} : \exists r_j \in \mathcal{R},~ \cov{u_i,r_j} \neq 0\}\) 
    \ForAll{\(u_i \in \mathcal{U}_{\mathcal{R}}\)} 
        \State \(u_i^{(r)} \gets \proj{\spanMath{\mathcal{R}}}{u_i}\) 
    \EndFor 
    \State \(\mathcal{U}_{\mathcal{R}}^{(r)} \gets \{u_i^{(r)} : u_i \in \mathcal{U}_{\mathcal{R}}\}\)
\end{algorithmic} 
\end{algorithm}

Specifically, step 4 in Algorithm~\ref{algo:construct_regressors} is performed by first denoting the vector of the \(q\) excitation signals as
\begin{equation}
	\bm{r}(t) = \trans{
	\begin{bmatrix}
		r_1(t) & \cdots & r_q(t)
	\end{bmatrix}
	} \in \mathbb{R}^q.
\end{equation}
For a selected maximum lag \(M\), we then define the regression vector as
\begin{equation}\label{eq:excitation_regressor}
	\bm{\phi}_r(t) = \trans{
	\begin{bmatrix}
		r_1(t) & \cdots & r_1(t - M) & r_2(t) & \cdots & r_q(t - M)
	\end{bmatrix}
	} \in \mathbb{R}^{q(M + 1)},
\end{equation}
in other words, lagged excitation samples are grouped first by excitation signal and then by lag, giving \(M + 1\) coefficients associated with each of the \(q\) excitations. The zero-lag term is included as we permit instantaneous excitation-to-input contributions. The recordings each contain \(N\) samples, and in order to ensure that all lagged values in \eqref{eq:excitation_regressor} are fully defined, we exclude the first \(M\) samples, leaving \(N - M\) usable data points for estimation. We then stack the regression vectors row-wise, giving the regression matrix
\begin{equation}
	\bm{\Phi}_r = \trans{
	\begin{bmatrix}
		\trans{\bm{\phi}_r(M)} & \trans{\bm{\phi}_r(M + 1)} & \cdots & \trans{\bm{\phi}_r(N - 1)}
	\end{bmatrix}
	} \in \mathbb{R}^{(N - M)\times(q(M + 1))},
\end{equation}
and collect the corresponding input measurements in
\begin{subequations}
\begin{align}
	\bm{u} &= \trans{
	\begin{bmatrix}
		\trans{\bm{u}(M)} & \trans{\bm{u}(M + 1)} & \cdots & \trans{\bm{u}(N - 1)}
	\end{bmatrix}
	} \in \mathbb{R}^{(N - M)\times m}~\text{where} \\
	\bm{u}(t) &= \trans{
	\begin{bmatrix}
		u_1(t) & \cdots & u_m(t)
	\end{bmatrix}
	} \in \mathbb{R}^m.
\end{align}
\end{subequations}
The \gls{fir} coefficients are then obtained by solving the multi-output least-squares problem
\begin{equation}\label{eq:least_squares}
	\hat{\bm{\Theta}} = \argmin_{\bm{\Theta}} \lVert\bm{\Phi}_r\bm{\Theta} - \bm{U}\rVert_\mathrm{F}^2
\end{equation}
where \(\hat{\bm{\Theta}} \in \mathbb{R}^{q(M + 1)\times m}\) is the vector of the estimated \gls{fir} coefficients and \(\lVert\cdot\rVert_\mathrm{F}\) denotes the Frobenius norm. Generally, the solution can be obtained using the pseudoinverse---denoted with \(\dagger\)---as
\begin{equation}
	\hat{\bm{\Theta}} = \bm{\Phi}_r^\dagger \bm{U},
\end{equation}
or if the regression matrix has full column rank, as
\begin{equation}
	\hat{\bm{\Theta}} = \inv*{\trans{\bm{\Phi}_r}\bm{\Phi}_r} \trans{\bm{\Phi}_r}\bm{U}.
\end{equation}
The projected inputs are then finally calculated as
\begin{equation}
	\tilde{\bm{U}} = \bm{\Phi}_r \hat{\bm{\Theta}} \in \mathbb{R}^{(N - M)\times m}.
\end{equation}
The projection residual is
\begin{equation}
	\bm{U}_\perp = \bm{U} - \tilde{\bm{U}},
\end{equation}
and by \eqref{eq:least_squares}, it satisfies \(\trans{\bm{\Phi}_r}\bm{U}_\perp = \bm{0}\). If the external excitation is independent of the disturbance, then, for a linear closed-loop system and a sufficiently flexible model, the projected inputs approximately retain the excitation-induced components while rejecting the process disturbance-induced component. Therefore,
\begin{equation}
	\tilde{u}_i(t) \approx u_i^{(r)}(t).
\end{equation}

Once suitable time-domain regressors have been obtained, the \gls{miso-gptf} estimation procedure continues as described in Section~\ref{sec:problem_formulation}, regardless of whether we use \(u_i\) or \(\tilde{u}_i\) in its place.
	\section{Results}\label{sec:results}
In the following section, we introduce a \gls{miso} power system case study, as well as perform \gls{mc} simulations to check the claims of Theorem~\ref{theo:exact_identifiability} and the closed-loop handling proposed in Section~\ref{sec:closed-loop_identifiability}.

\subsection{Dynamic Power System Model}\label{sec:power_system}
Our test case describes the simplified frequency-active power dynamics of a two-bus power system where two generating units are connected through a high voltage AC transmission line, neglecting losses. Let
\begin{equation}
	\bm{Y}(s) = \trans{
	\begin{bmatrix}
		Y_1(s) & Y_2(s)
	\end{bmatrix}
	}
\end{equation}
denote the vector of local per-unit angular frequency deviations, and let
\begin{equation}
	\bm{U}(s) = \trans{
	\begin{bmatrix}
		U_1(s) & U_2(s)
	\end{bmatrix}
	}
\end{equation}
denote the vector of controllable power injections. \(s\) is the Laplace variable. The inputs and outputs at bus 1 are then related by the linearized swing dynamics
\begin{subequations}\label{eq:swings}
\begin{equation}\label{eq:swing1}
    sM_1 Y_1(s) = U_1(s) - D_1Y_1(s) - P_{12}(s)
\end{equation}
where \(M_1\) and \(D_1\) are, respectively, inertia and damping coefficients and \(P_{12}\) is the incremental active power transferred from bus 1 to bus 2. The corresponding dynamics for bus 2 are
\begin{equation}\label{eq:swing2}
    sM_2 Y_2(s) = U_2(s) - D_2Y_2(s) + P_{12}(s)
\end{equation}
\end{subequations}
where \(M_2\) and \(D_2\) are again inertia and damping coefficients. The opposite signs of \(P_{12}\) in \eqref{eq:swing1} and \eqref{eq:swing2} follow from the selected positive power flow direction.

For a lossless transmission line and small deviations around the operating point, the incremental power change can be expressed as
\begin{equation}\label{eq:power_transfer}
    P_{12}(s) = \frac{K_{12}}{s}(Y_1(s) - Y_2(s))
\end{equation}
where \(K_{12} > 0\) is the synchronizing coefficient. Substituting \eqref{eq:power_transfer} into \eqref{eq:swings} gives us the system representation
\begin{equation}\label{eq:inv_swing}
	\begin{bmatrix}
		U_1(s) \\
		U_2(s)
	\end{bmatrix}
	=
	\begin{bmatrix}
		sM_1 + D_1 + K_{12}/s & -K_{12}/s \\ 
		-K_{12}/s & sM_2 + D_2 + K_{12}/s \\ 
	\end{bmatrix}
	\begin{bmatrix}
		Y_1(s) \\
		Y_2(s)
	\end{bmatrix}.
\end{equation}
Equation~\eqref{eq:inv_swing} describes the transfer from outputs to inputs. We therefore invert the transfer matrix in \eqref{eq:inv_swing} and get the disturbance-free power-to-frequency dynamics
\begin{subequations}\label{eq:swings_matrix}
\begin{align}
    \bm{Y}_0(s) &= \bm{\Psi}(s)\bm{U}(s)~\text{where}\\
	\bm{\Psi}(s) &= \frac{1}{\alpha(s)}
	\begin{bmatrix}
		s^2M_2 + sD_2 + K_{12} & K_{12} \\
		K_{12} & s^2M_1 + sD_1 + K_{12}
	\end{bmatrix}~\text{and} \\
	\begin{split}
		\alpha(s) &= s^3(M_1M_2) + s^2(M_1D_2 + M_2D_1) \\
		&\qquad + s(M_1K_{12} + M_2K_{12} + D_1D_2) + (D_1 + D_2)K_{12}.
	\end{split}
\end{align}
\end{subequations}
\(\bm{Y}_0(s)\) is the disturbance-free output vector. Therefore, in order to agree with the model formulation introduced in Section~\ref{sec:problem_formulation}, the observed outputs are modeled as
\begin{subequations}\label{eq:system_eq}
\begin{align}
    \bm{Y}(s) &= \bm{\Psi}(s)\bm{U}(s) + \bm{V}(s)~\text{where} \\
    \bm{V}(s) &= \trans{
	\begin{bmatrix}
		V_1(s) & V_2(s)
	\end{bmatrix}
	}.
\end{align}
\end{subequations}
\(\bm{V}\) is an equivalent output disturbance expressed in the same units as \(\bm{Y}\).

We can close the loop by introducing local feedback controllers
\begin{equation}
	\bm{C}(s) =
	\begin{bmatrix}
		C_1(s) & 0 \\
		0 & C_2(s)
	\end{bmatrix},
\end{equation}
where each controller \(C_1(s)\) and \(C_2(s)\) takes in the difference between an exogenous reference/excitation signal
\begin{equation}
	\bm{R}(s) = \trans{
	\begin{bmatrix}
		R_1(s) & R_2(s)
	\end{bmatrix}
	}
\end{equation}
and the output \(\bm{Y}\), resulting in the control equation
\begin{equation}\label{eq:control_eq}
	\bm{U}(s) = \bm{C}(s)(\bm{R}(s) - \bm{Y}(s)).
\end{equation}

At bus 1, the controller is represented by the governor-turbine model
\begin{equation}
	C_1(s) = \frac{K_\mathrm{g}}{(1 + sT_\mathrm{g})(1 + sT_\mathrm{t})},
\end{equation}
where \(T_\mathrm{g}\) and \(T_\mathrm{t}\) are, respectively, the equivalent governor and turbine time constants, and \(K_\mathrm{g}\) is the governor gain. The signal \(U_1\) is then the incremental mechanical power supplied by the generator. We let the bus 2 dynamics represent a grid-forming inverter, such that \(M_2\) and \(D_2\) are virtual inertia and damping, and the controller is a simple fast filtered droop controller, specifically
\begin{equation}
	C_2(s) = \frac{K_\mathrm{d}}{(1 + sT_\mathrm{c})}.
\end{equation}
\(K_\mathrm{d}\) is the inverter gain and \(T_\mathrm{c}\) is the filter time constant. \(U_2\) is then the incremental active power injection by the inverter. The values used for the power system case are given in \ref{apx:sim_params}.

Combining \eqref{eq:system_eq} and \eqref{eq:control_eq}, as well as dropping the explicit dependencies on \(s\) in order to ease notation, gives the closed-loop equation
\begin{subequations}
\begin{align}
    \bm{Y}(s) &= \bm{\Psi}\bm{C}(\bm{R} - \bm{Y}) + \bm{V} \\
    &= \inv*{\bm{I}_2 + \bm{\Psi}\bm{C}}\bm{\Psi}\bm{C}\bm{R} + \inv*{\bm{I}_2 + \bm{\Psi}\bm{C}}\bm{V}.
\end{align}
\end{subequations}
The measured structural inputs are
\begin{equation}
    \bm{U}(s) = \inv*{\bm{I}_2 + \bm{C}\bm{\Psi}}\bm{C}\bm{R} - \inv*{\bm{I}_2 + \bm{C}\bm{\Psi}}\bm{C}\bm{V},
\end{equation}
which corresponds directly with \eqref{eq:input_decomp}.

We designate \(Y_1\) as the single output of interest, reducing the power system model to a \gls{miso} system. Specifically, the first row of \eqref{eq:swings_matrix} gives
\begin{subequations}
\begin{align}
    Y_1(s) &= G_1(s)U_1 + G_2(s)U_2 + V_1~\text{with} \\
    G_1(s) &= \frac{s^2M_2 + sD_2 + K_{12}}{\alpha(s)}~\text{and} \\
    G_2(s) &= \frac{K_{12}}{\alpha(s)}.
\end{align}
\end{subequations}
\(G_1\) and \(G_2\) are the \glspl{tf} we aim to identify. 

We consider two configurations of the system: one where the controllers are active---preserving the closed-loop nature of the system---and the structural inputs are generated according to \eqref{eq:control_eq}, and one open-loop configuration where the feedback paths are disconnected and \(U_1\) and \(U_2\) are applied directly. Doing so lets us separately investigate the two independent problems of \gls{miso-gptf} estimation, namely fulfilling the requirements in Theorem~\ref{theo:exact_identifiability} and constructing valid regressors for closed-loop systems. An overview of the test system is shown in Fig.~\ref{fig:test_system}.
\begin{figure}[tb]
	\centering
	\resizebox{\linewidth}{!}{
	\ctikzset{bipoles/thickness=0.75}
	\tikzstyle{block} = [draw, rectangle, minimum height=1.5cm, text width=2.5cm, line width=0.7pt, align=center]
	\tikzstyle{smallblock} = [draw, rectangle, minimum height=0.75cm, text width=1.5cm, line width=0.7pt, align=center]
	\tikzstyle{miniblock} = [draw, rectangle, minimum height=0.75cm, text width=0.75cm, line width=0.7pt, align=center]
	\tikzstyle{branch} = [circle, inner sep=0pt, minimum size=1mm, fill=black, draw=black]
	\tikzstyle{connector} = [->, line width=0.7pt]
	\tikzstyle{line} = [line width=0.7pt]
	\tikzstyle{dummy} = [inner sep=0pt, minimum size=0pt]
	\tikzstyle{inout} = []
	\tikzstyle{sum} = [circle, inner sep=0pt, minimum size=2mm, draw=black, line width=0.7pt]
	\tikzstyle{product} = [rectangle, inner sep=0pt, minimum size=10mm, draw=black, line width=0.7pt]
	
	\subfloat[]{
	\begin{circuitikz}[line width=0.75, >=stealth, node distance=1cm, every node/.append style={transform shape}, american]
  		\tikzstyle{every node}=[font=\Large]
		
		\node[inout] (u1) {$U_1$};
    	\node[block] (g1) [right=of u1, color=blue] {$G_1(s)$};
    	\node[sum] (s1) [right=of g1] {};
    	\node[sum] (s2) [right=of s1] {};
    	\node[inout] (v1) [below=of s2] {$V_1$};
    	\node[inout] (y) [right=of s2] {$Y$};
        
    	\node[block] (g2) [below=of g1, color=blue] {$G_2(s)$};
    	\node[inout] (u2) [left=of g2] {$U_2$};
        
        \draw[connector] (u1) -- (g1);
        \draw[connector] (g1) -- (s1);
        \draw[connector] (s1) -- (s2);
        \draw[connector] (s2) -- (y);
        \draw[connector] (v1) -- (s2);
        \draw[connector] (u2) -- (g2);
        \draw[connector] (g2) -| (s1);
	\end{circuitikz}}
	\hfill
  	\subfloat[]{
  	\begin{circuitikz}[line width=0.75, >=stealth, node distance=1cm, every node/.append style={transform shape}, american]
		\tikzstyle{every node}=[font=\Large]
		
		\node[inout] (r1) {$R_1$};
        \node[sum] (s1) [right=of r1] {};
        \node[dummy] (d1s1) [above right=0.05cm and 0.05cm of s1] {-};
        \node[block] (c1) [right=of s1] {$C_1(s)$};
        \node[branch] (br1) [right=of c1] {};
        \node[block, color=blue] (g1) [right=of br1] {$G_1(s)$};
        \node[sum] (s2) [right=of g1] {};
        \node[sum] (s3) [right=of s2] {};
        \node[inout] (v1) [below=of s3] {$V_1$};
        \node[branch] (br2) [right=of s3] {};
        \node[inout] (y) [right=of br2] {$Y_1 = Y$};

        \node[block, color=blue] (g12) [below=of g1] {$G_2(s)$};
        \node[block] (g21) [below=of g12] {$\psi_{21}(s)$};

        \node[block] (g2) [below=of g21] {$\psi_{22}(s)$};
        \node[dummy] (d1) [left=of g2] {};
        \node[branch] (br4) [left=0.45 of d1] {};
        \node[block] (c2) [left=of d1] {$C_2(s)$};
        \node[sum] (s4) [left=of c2] {};
        \node[dummy] (d1s4) [below right=0.05cm and 0.05cm of s4] {-};
        \node[inout] (r2) [left=of s4] {$R_2$};
        \node[sum] (s5) [right=of g2] {};
        \node[sum] (s6) [right=of s5] {};
        \node[inout] (v2) [above=of s6] {$V_2$};
        \node[dummy] (d1s6) [above right=0.05cm and 0.5cm of s6] {$Y_2$};
    
        \draw[connector] (r1) -- (s1);
        \draw[connector] (s1) -- (c1);
        \draw[line] (c1) -- (br1);
        \draw[connector] (br1) -- (g1) node[pos=0, above] {$U_1$};
        \draw[connector] (g1) -- (s2);
        \draw[connector] (s2) -- (s3);
        \draw[line] (s3) -- (br2);
        \draw[connector] (br2) -- (y);
        \draw[connector] (v1) -- (s3);
        \draw[connector] (br2) |- ++(0, 1.5) -| (s1);
        \draw[connector] (br1) |- (g21);
        \draw[connector] (g12) -| (s2);

        \draw[connector] (r2) -- (s4);
        \draw[connector] (s4) -- (c2);
        \draw[connector] (c2) -- (g2) node[midway, above] {$U_2$};
        \draw[connector] (g2) -- (s5);
        \draw[connector] (s5) -- (s6);
        \draw[connector] (v2) -- (s6);
        \draw[connector] (s6) -- ++(1.175, 0) |- ++(0, -1.5) -| (s4);
        \draw[connector] (br4) |- (g12);
        \draw[connector] (g21) -| (s5);
	\end{circuitikz}}
	}
  \caption{\gls{miso} test system modeling the frequency-active power dynamics of a two-bus power system. (a) is the open-loop configuration where \(U_1\) and \(U_2\) are exogenous and applied directly, and (b) is the closed-loop configuration where \(R_1\) and \(R_2\) are applied externally, while \(U_1\) and \(U_2\) become endogenous signals. The \glspl{tf} we aim to identify are \textcolor{blue}{\(G_1\)} and \textcolor{blue}{\(G_2\)}, and \(\psi_{21}\) and \(\psi_{22}\) are the second-row elements in \eqref{eq:swings_matrix}.}
	\label{fig:test_system}
\end{figure}
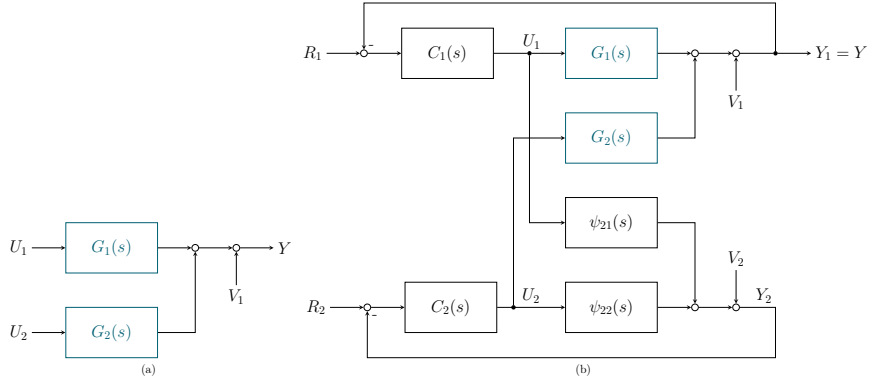

\subsection{Simulation Setup}
For both system configurations, we set up simulations with sampling period \(T_\mathrm{s} = \SI{0.02}{\second}\) and measurement length \(T = \SI{300}{\second}\). The external excitations \(r_i\) and disturbance signals \(v_i\) are all zero-mean band-limited white noise with distinct realizations. After having applied the \gls{dft}, the upper angular frequency of interest is set to \SI{10}{\radian\per\second} and the same \(F = 475\) logarithmically-spaced angular frequency points are used for both training and testing\footnote{Depending on the needs of the user, the training and testing frequencies need not coincide. See~\cite{NarayanEtAl2026} for further details.}.

For the open-loop tests, we vary how similar the two inputs are in order to generate input-output tuples that either pass or fail the exact identifiability condition. Specifically, we set up two independent inputs \(u_1\) and \(u_\perp\), and generate input \(u_2\) from
\begin{equation}
	u_2(t) = \beta u_1(t) + \sqrt{1 - \beta^2}u_\perp(t).
\end{equation}
\(u_\perp\) is a dummy input that is only there to provide \(u_2\) with content that is more or less similar to \(u_1\), depending on the parametrization of \(\beta\in [0, 1]\). \(u_\perp\) and \(u_1\) are independent zero-mean band-limited white noise signals. When \(\beta = 0\), \(u_1\) and \(u_2\) are completely independent, and when \(\beta = 1\), they are identical. We run simulations for \(\beta \in \{0.0, 0.5, 0.9, 0.99, 1.0\}\). For each of the input-output sets, we then perform \gls{miso-gptf} estimation, evaluate the exact identifiability condition, and compare the module estimate to the true module using the fit
\begin{equation}\label{eq:fit}
    \rho_i = \maxMath{1 - \frac{\lVert \hat{\bm{G}}_i - \bm{G}_i \rVert_2}{\left\lVert \bm{G}_i \right\lVert_2}, 0},\quad \left\| \bm{G}_i \right\|_2 > 0
\end{equation}
where \(\left\lVert \cdot \right\rVert_2\) is the \(\mathrm{L}^2\)-norm.

For the closed-loop tests, we vary the \gls{snr}, defined as the ratio between the variances of the excitation- and disturbance-contributions to input 1. For each \gls{snr}, we compare the module fits achieved by using the endogenous structural inputs as regressors, to the fit achieved by projecting the inputs on the exogenous excitation signals. We test with \glspl{snr} \qtylist[list-units = single]{20; 10; 6; 3; 0}{\decibel}, all with independent excitation signals. 

We perform 100 \gls{mc} runs for each open- and closed-loop test where the random exogenous signals are regenerated for each run. Throughout, we use the diagonal correlated covariance function for modules 0 (the transient) and 2, and the stable spline covariance function for module 1. References~\cite{LataireChen2016, NarayanEtAl2026} contain further analysis of covariance function choice for frequency-domain estimation. Note that we only perform hyperparameter optimization for the open-loop case where \(\beta = 0\) and then reuse the hyperparameters for the subsequent input correlation and \gls{snr} sweeps. Doing so allows the \gls{gp} prior to adapt to each random realization while ensuring that changes in the evaluation metrics can be attributed to the controlled experimental variable, rather than case-specific changes in the prior-supported module spaces.

\subsection{Open-Loop Identifiability Results}
Fig.~\ref{fig:ol_fits} shows the empirical distribution functions for module 1 and 2 fits, separated into \gls{mc} runs that do and do not pass the exact identifiability requirement. For the cases that fulfill the requirement, the fits are strongly concentrated near the ideal 1.0 from below, for both modules. For the cases that do not pass, more than half have module 1 fits at 0.0, with a considerable jump at around 0.2. The failed module 1 fits all fall below a 0.65 fit. The failed module 2 fits are overall stronger with a median value just under 0.8. Still, the maximum achieved failed module 2 fit is 0.94---below the minimum achieved fit for any of the cases that pass---giving a meaningful fit separation between the pass and fail runs.
\begin{figure}[tb]
    \centerline{\includegraphics[width=\linewidth]{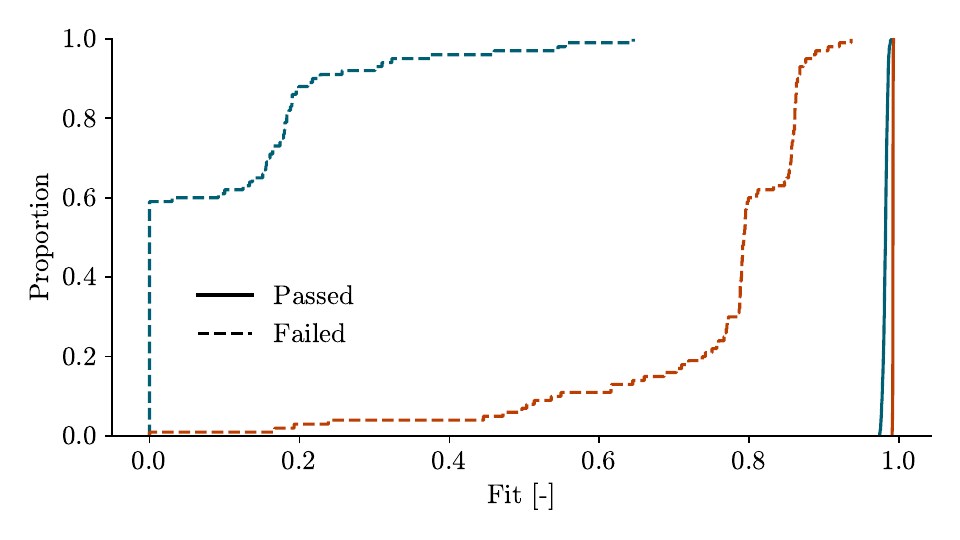}}
    \caption{Empirical distribution functions of \textcolor{blue}{module 1} and \textcolor{orange}{module 2} fits, split into cases that fulfill and cases that do not fulfill the exact identifiability criterion. Among the \gls{mc} runs, cases that pass the requirement all have close to ideal fits, while the cases that do not pass the criterion have considerably worse fits.}
    \label{fig:ol_fits}
\end{figure}

A different perspective can be obtained by considering Fig.~\ref{fig:ol_module_fit_comparison} where the module 2 fits are plotted against the module 1 fits. It becomes clear that in order to have satisfactory fits for both modules, the rank requirement must be fulfilled. A considerable number of the \gls{mc} runs resulted in module 2 fits around 0.8 while failing the rank requirement. However, these runs struggled to achieve module 1 fits beyond 0.3, with major clusters at 0.0 and around 0.2. 
\begin{figure}[tb]
    \centerline{\includegraphics[width=\linewidth]{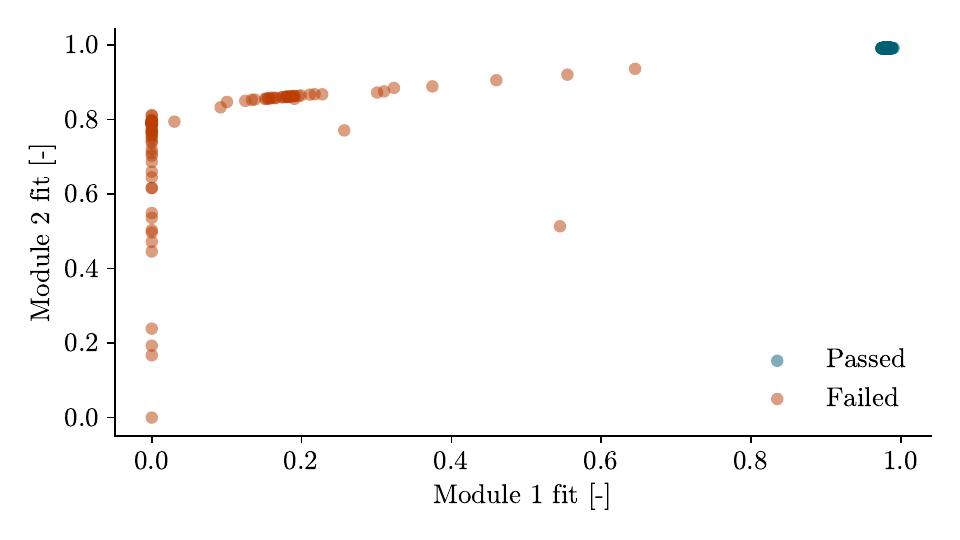}}
    \caption{Scatter plot of module 2 fit against module 1 fit from \gls{mc} runs. Module 2 fits are generally stronger than module 1 fits, but close to ideal fits for both modules is only achieved when the exact identifiability condition is satisfied.}
    \label{fig:ol_module_fit_comparison}
\end{figure}

A further illustration of the contrast between passing and not passing the exact identifiability criterion is given in Fig.~\ref{fig:ol_bode} where the \gls{miso-gptf} estimates of representative \gls{mc} runs are shown. The passing setup with \(\beta = 0\) achieves fits of 0.99 for both modules, while the setup with \(\beta = 1\) results in fits of 0.0 for module 1 and 0.79 for module 2. Along with the module estimates, Fig.~\ref{fig:ol_bode} also includes posterior samples of all modules, including the transient 0-module, illustrating the uncertainty in the posterior estimates. Fig.~\ref{fig:ol_bode} highlights not only a clear difference in the accuracy of the estimation, depending on the rank requirement being fulfilled, but also how the posterior uncertainty shrinks considerably in the passed setup, compared to the failed. In the failed case, the module 1 estimate is consistently poor across the considered frequency range, while the module 2 estimate matches the true frequency response quite well up until about \SI{1e0}{\radian\per\second} before deteriorating. The behavior is consistent with the dynamics of the test system: both modules share an oscillatory mode near \SI{1.5}{\radian\per\second}, beyond which the \gls{tf} magnitudes decrease and their contributions to the measured output become progressively weaker. Consequently, the effective \glspl{snr} decrease, while disturbance, transients, and prior regularization become relatively more influential.
\begin{figure}[tb]
  \centering
  \subfloat[]{
    \includegraphics[width=0.9\linewidth]{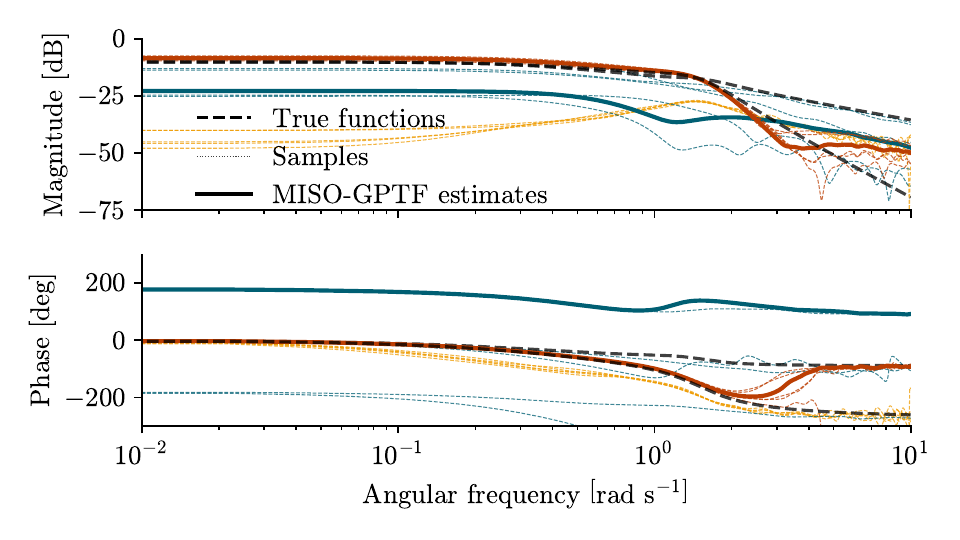}
    \label{fig:ol_bode_failed}}
  \hfill
  \subfloat[]{
    \includegraphics[width=0.9\linewidth]{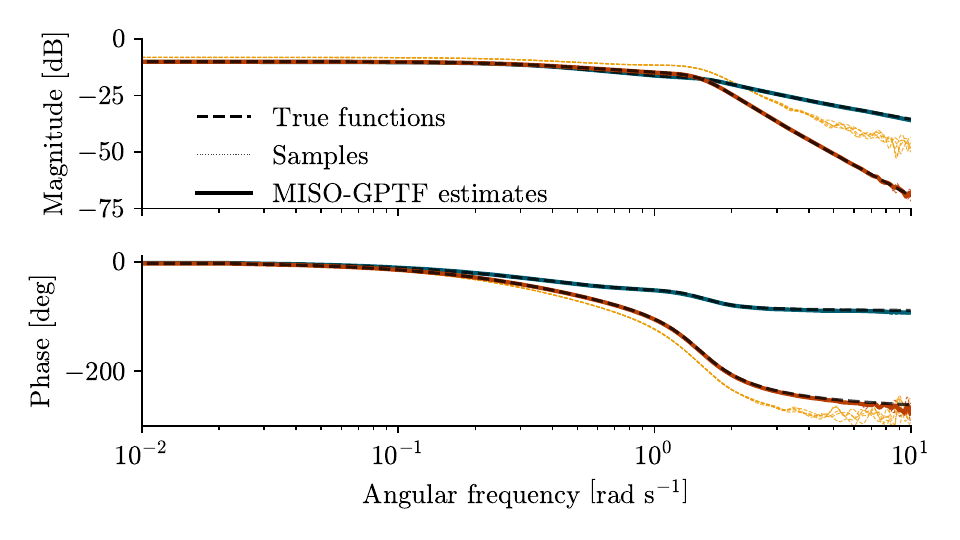}
    \label{fig:ol_bode_passed}}
  \caption{Example of \gls{miso-gptf} estimates along with true frequency responses and posterior samples for \textcolor{blue}{module 1} and \textcolor{orange}{module 2}. Posterior samples of \textcolor{yellow}{module 0 (transient term)} are also included. (a) shows estimated \glspl{tf} from when the exact identifiability requirement is not met, resulting in poor fits, and (b) shows the resulting well-fitted estimated \glspl{tf} from when the requirement is met. Note the difference in phase axis limits between (a) and (b).}
  \label{fig:ol_bode}
\end{figure}

\subsection{Closed-Loop Projection Results}
Turning our attention to Fig.~\ref{fig:cl_snr_sweep}, we see that, regardless of whether we use the projected or direct inputs, the fit of module 2 is consistently better than the module 1 fit. The discrepancy is reduced as the \gls{snr} increases before it all but disappears at \SI{20}{\decibel}. Still, the achieved fit is strongly dependent on \gls{snr}---even more so for module 1 than for module 2. Continuing on to the projection advantage, we observe the strongest effect for lower \glspl{snr}. One possible explanation is that as the \gls{snr} increases, the \gls{fir} approximation error becomes relatively more dominant compared to the closed-loop bias and we get a reduced or even detrimental effect from projecting. Moreover, module 2 fit has a slight concave trend, peaking at \gls{snr} \SI{10}{\decibel}, but the effect of projecting is consistently negative. There are \gls{mc} runs where the advantage is positive, but it is minimal and the median is consistently below zero. The effect of projecting is therefore channel-dependent, showing that constructing exogenous regressors does not automatically improve all module estimates.
\begin{figure}[tb]
  \centering
  \subfloat[]{
    \includegraphics[width=0.9\linewidth]{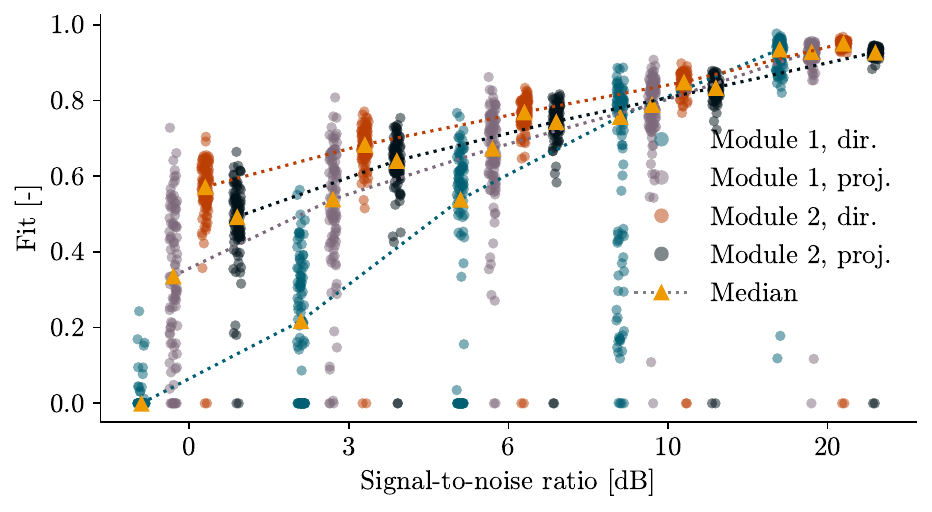}
    \label{fig:cl_fits}}
  \hfill
  \subfloat[]{
    \includegraphics[width=0.9\linewidth]{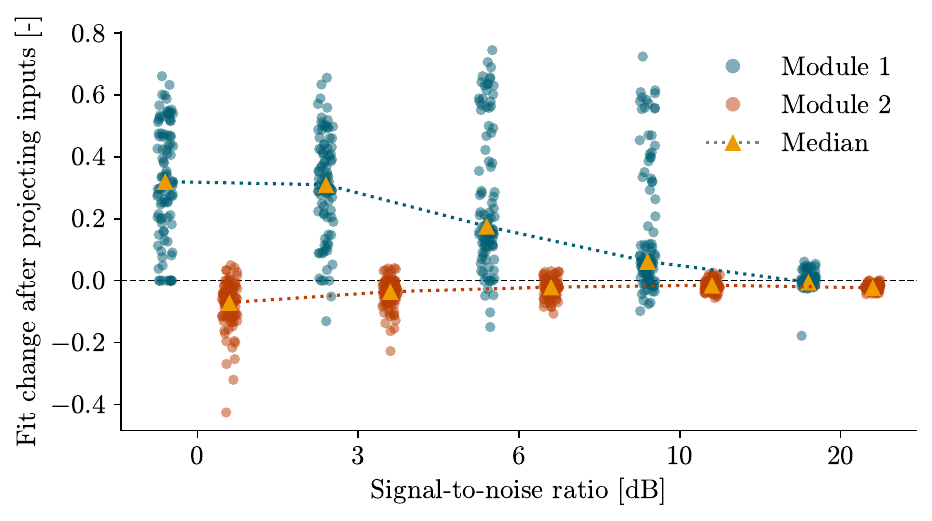}
    \label{fig:cl_fits_change}}
  \caption{Effect of input projection on closed-loop identification performance across \gls{snr} levels. (a) shows the module fits obtained with direct and projected inputs, and (b) the change in module fits after projecting. The median values across the \gls{mc} runs are included as well. Increasing \gls{snr} is correlated with increasing fits and decreasing fit differences between using direct and projected inputs. The effects are stronger for \textcolor{blue}{module 1} than for \textcolor{orange}{module 2}. Projecting provides the greatest improvements for \textcolor{blue}{module 1} at low \gls{snr}, but the benefit decreases as \gls{snr} increases. For \textcolor{orange}{module 2}, projecting has a small yet consistently negative median effect.}
  \label{fig:cl_snr_sweep}
\end{figure}
\FloatBarrier
	\section{Discussion}\label{sec:discussion}
Despite the presented work demonstrating the validity of the exact identifiability criterion, several open questions and limitations remain.

Firstly, the open-loop module fits hint that the rank requirement of exact identifiability may be too strict, meaning the \gls{miso-gptf} estimator could potentially achieve acceptable fits---possibly only in limited frequency ranges---even when failing the requirement. Despite the inputs being different and not passing exact identifiability, the module priors being different enough may still make it possible to separate and identify each module. For practical \gls{miso} identification, introducing the notion of stable identifiability and deriving an identifiability requirement in terms of singular values which would allow a more nuanced consideration of when the \gls{miso-gptf} estimator is expected to fail---as well as quantify how prior- or data-dependent the estimates are---is a potentially fruitful path.

Secondly, the results highlight that module separability and regressor selection are distinct problems. A full-rank design can still yield poor closed-loop estimates if the regressors are endogenous and contaminated. Conversely, exogenous regressors cannot recover modules whose contribution spaces are not separable. Still, despite the derivation of the estimator assuming independence between the regressors and disturbances, the results are consistent with a trade-off between the closed-loop bias and approximation and estimation errors introduced when constructing projected regressors. In practical terms, the bias may very well be acceptable, given a sufficiently high \gls{snr}, where the benefit of projecting is reduced. Moreover, the deviating results between the modules also suggest that each channel should be treated more individually. This motivates an adaptive or iterative approach where we first start with only direct regressors, before projecting or otherwise constructing corresponding regressors channel-wise.

Moreover, throughout this work, from the choice of covariance functions to the derivations of the \gls{miso-gptf} convergence, the module estimates are treated as independent. However, if similarities are known or suspected a priori, a covariance function encoding relationships between the modules could improve the estimation by allowing the estimator to borrow information between the modules using coregionalization matrices~\cite{BonillaEtAl2007}.
	\section{Conclusions}\label{sec:conclusions}
The preceding work has presented a data-driven approach to identifying closed-loop \acrlong{miso} systems using \acrlongpl{gp}, extending existing \acrlong{siso} approaches. The approach results in continuous frequency-domain \acrlongpl{tf}, along with quantified uncertainty bounds. In addition to formulating the estimator, we have derived under what conditions the separate modules are exactly identifiable, showing that the measured input spectra and covariance function choice play a key role. Moreover, we have shown how conventional closed-loop bias handling can be employed to extend the range of problems to which the estimator can be applied by adjusting the regressors used. Through \acrlong{mc} simulations we demonstrated that the theoretical identifiability requirement strongly predicts the fit of the system modules. 

For the closed-loop case studies, we observed a trade-off between the closed-loop bias associated with using endogenous inputs and the error introduced in constructing projected exogenous regressors. Projecting improved the estimates of module 1 most clearly at low \gls{snr}, while its advantage decreased with increasing \gls{snr}. For module 2, on the other hand, projecting had a small but consistently detrimental effect on model fit. The results indicate that regressor construction should be treated on a per-channel basis, rather than as a default preprocessing step.

Further work will focus on refining the developed binary pass/fail requirement into a continuous metric in conjunction with robust regressor selection and handling.
	
	\section*{Acknowledgments}
This work is supported by the Research Council of Norway and Statkraft Energi AS through project number 346334.

	\appendix
	\section{Proof of Theorem~\ref{theo:exact_identifiability}}\label{apx:exact_identifiability}
The proof consists of two independent stages:
\begin{enumerate}[label=(\arabic*)]
	\item Disentangling the outputs. Given an observation \(Y = \sum_i Z_i\), can each contribution \(Z_i\) be uniquely determined from \(Y\)?
	\item Inverting each contribution. Once \(Z_i\) is known, can we uniquely solve \(\mathcal{M}_iG_i = Z_i\) for \(G_i\)?
\end{enumerate}
Only when both are satisfied can we uniquely recover the original transfer function tuple \((G_0,\dots, G_m)\). 

First consider contribution identifiability. We will show that the contributions forming a direct sum and being unique are equivalent. Suppose \(\mathcal{S}_0 + \dots + \mathcal{S}_m\) is a direct sum. If
\begin{equation}
	Y = Z_0 + \dots + Z_m = \tilde{Z}_0 + \dots + \tilde{Z}_m, \quad Z_i, \tilde{Z}_i \in \mathcal{S}_i,
\end{equation}
i.e., there are two ways of constructing a given \(Y\), then
\begin{equation}
	0 = (Z_0 - \tilde{Z}_0) + \dots + (Z_m - \tilde{Z}_m).
\end{equation}
Since \(\mathcal{S}_i\) is a linear subspace, \(Z_i - \tilde{Z}_i \in \mathcal{S}_i\), and directness implies
\begin{equation}
	0 = Z_i - \tilde{Z}_i \quad \implies \quad Z_i = \tilde{Z}_i, \quad i = 0, \dots, m.
\end{equation}
In other words, each contribution \(Z_i\) is unique.
	
Conversely, suppose each contribution is unique. To show that their sum is then direct, suppose there is an arbitrary collection of \(s_i \in \mathcal{S}_i\) satisfying \(s_0 + \dots + s_m = 0\). The zero output then has two decompositions:
\begin{equation}
	0 = 0 + \dots + 0 \quad \text{and} \quad 0 = s_0 + \dots + s_m.
\end{equation}
Uniqueness then implies that the corresponding contributions are equal, so \(s_i = 0,~i = 0, \dots, m\), and the sum is therefore direct.
	
We now move on to module identifiability. In order to recover the original modules, in addition to the direct sum condition, every \(\mathcal{M}_i\) must be injective. If \(\mathcal{M}_i\) fails to be injective, there exist distinct transfer functions \(G_i \neq \tilde{G}_i\) with \(\mathcal{M}_iG_i = \mathcal{M}_i\tilde{G}_i = Z_i\), making it impossible to know which transfer function gave rise to the output contribution.
	
We will show that injectivity of \(\mathcal{A}\) implies injectivity of each \(\mathcal{M}_i\) and that the module spaces \(\mathcal{S}_i\) form a direct sum. Starting with necessity and assuming \(\mathcal{A}\) is injective, pick an arbitrary module \(j \in \{0,\dots, m\}\) and suppose that \(\mathcal{M}_jG_j = 0\). For \(\mathcal{M}_j\) to be injective, the only transfer function that can be mapped to zero under \(\mathcal{M}_j\) is the zero function. Consider the tuple whose \(j\)th component is \(G_j\) and whose remaining components are zero. Using \eqref{eq:observation-operator} gives
\begin{equation}
	\mathcal{A}(0, \dots, 0, G_j, 0, \dots, 0) = \mathcal{M}_jG_j = 0.
\end{equation}
This tuple then maps to zero under \(\mathcal{A}\). However, we assumed \(\mathcal{A}\) to be injective, so the only tuple \(\mathcal{A}\) maps to zero is the zero tuple. Therefore, for an arbitrary \(j\), \(G_j = 0\) if \(\mathcal{A}\) is injective. This again makes \(\mathcal{M}_j\) injective since only \(G_j = 0\) makes \(\mathcal{M}_jG_j = 0\). 

Next, we show that \(\mathcal{A}\) being injective implies that the \(\mathcal{S}_i\) form a direct sum. Suppose we have a collection of \(s_i \in \mathcal{S}_i\) which satisfies \(s_0 + \dots + s_m = 0\). By the definition of \(\mathcal{S}_i\), there exists \(G_i \in \mathcal{H}_i\) such that
\begin{equation}
	s_i = \mathcal{M}_iG_i.
\end{equation}
Substituting then gives
\begin{equation}\label{eq:sum-of-si}
	0 = s_0 + \dots + s_m = \mathcal{M}_0G_0 + \dots + \mathcal{M}_mG_m
	= \mathcal{A}(G_0, \dots, G_m).
\end{equation}
We assumed \(\mathcal{A}\) to be injective, meaning that \((G_0, \dots, G_m)\) must be the zero tuple. Since \(s_i = \mathcal{M}_iG_i\) and \(G_i = 0\), we get
\begin{equation}
	s_i = \mathcal{M}_i 0 = 0,
\end{equation}
meaning that the \(\mathcal{S}_i\) form a direct sum. The equality follows from the linearity of \(\mathcal{M}_i\).
	
Having proven necessity, we now move on to proving sufficiency, i.e., assuming that each \(\mathcal{M}_i\) is injective and that the \(\mathcal{S}_i\) form a direct sum, showing that this means \(\mathcal{A}\) is injective. Say we have a tuple \((G_0, \dots, G_m)\) that maps to zero under \(\mathcal{A}\), giving us the same situation as in \eqref{eq:sum-of-si}. Having assumed that the \(\mathcal{S}_i\) form a direct sum, \(s_0 + \dots + s_m = 0\) means that \(s_0 = \dots = s_m = 0\) and therefore that \(\mathcal{M}_iG_i = 0,~i = 0, \dots, m\). Having also assumed \(\mathcal{M}_i\) is injective, only \(G_i = 0\) would make \(\mathcal{M}_iG_i = 0\). The tuple \((G_0, \dots, G_m)\) is therefore the zero tuple and \(\mathcal{A}\) is injective, concluding the proof.\qed
	\section{Proof of Proposition~\ref{prop:miso-gptf_consistency}}\label{apx:miso-gptf_consistency}
We begin by inserting \eqref{eq:prewhitened_model} into \eqref{eq:post_latent_mean}, giving
\begin{equation}\label{eq:post_latent_mean_inter}
	\hat{\bm{W}} = C_\sigma\left(\bm{L}^{(W)}\bm{W} + \sigma\bm{\epsilon}\right),
\end{equation}
where
\begin{equation}
    C_\sigma = \trans*{\bm{L}^{(W)}}\inv*{\bm{L}^{(W)}\trans*{\bm{L}^{(W)}} + \sigma^2\bm{I}_{2F}}.
\end{equation}
In order to simplify further, it is useful to have with us that \eqref{eq:svd} and matrix manipulation gives
\begin{equation}
	\bm{L}^{(W)}\trans*{\bm{L}^{(W)}} = \bm{U}_\mathrm{s}\bm{S}^2\trans{\bm{U}_\mathrm{s}},
\end{equation}
as well as
\begin{subequations}
\begin{align}
	 C_\sigma &= \left(\bm{V}_\mathrm{s}\bm{S}\trans{\bm{U}_\mathrm{s}}\right)\inv*{\bm{U}_\mathrm{s}\bm{S}^2\trans{\bm{U}_\mathrm{s}} +\sigma^2\bm{I}_{2F} + \sigma^2\bm{U}_\mathrm{s}\trans{\bm{U}_\mathrm{s}} - \sigma^2\bm{U}_\mathrm{s}\trans{\bm{U}_\mathrm{s}}} \\
	 &= \left(\bm{V}_\mathrm{s}\bm{S}\trans{\bm{U}_\mathrm{s}}\right)\inv*{\bm{U}_\mathrm{s}\left(\bm{S}^2 + \sigma^2\bm{I}_P\right)\trans{\bm{U}_\mathrm{s}} + \sigma^2\left(\bm{I}_{2F} - \bm{U}_\mathrm{s}\trans{\bm{U}_\mathrm{s}}\right)} \\
	 &= \left(\bm{V}_\mathrm{s}\bm{S}\trans{\bm{U}_\mathrm{s}}\right)\left(\bm{U}_\mathrm{s}\inv*{\bm{S}^2 + \sigma^2\bm{I}_P}\trans{\bm{U}_\mathrm{s}} + \sigma^{-2}\left(\bm{I}_{2F} - \bm{U}_\mathrm{s}\trans{\bm{U}_\mathrm{s}}\right)\right) \\
	 &= \bm{V}_\mathrm{s}\bm{S}\inv*{\bm{S}^2 + \sigma^2\bm{I}_P}\trans{\bm{U}_\mathrm{s}} \\
	 &= \bm{V}_\mathrm{s}\diag{\frac{s_p}{s_p^2 + \sigma^2}}\trans{\bm{U}_\mathrm{s}}.
\end{align}
\end{subequations}
We have used that \(\bm{U}_\mathrm{s}\trans{\bm{U}_\mathrm{s}}\) and \(\bm{I}_{2F} - \bm{U}_\mathrm{s}\trans{\bm{U}_\mathrm{s}}\) are mutually orthogonal projection matrices that sum to the identity. They decouple the vector space into independent subspaces, so we can find the overall inverse by simply inverting the weights associated with each individual projection piece. Returning to \eqref{eq:post_latent_mean_inter}, this lets us express
\begin{align}\label{eq:latent_mean_svd_apx}
	\hat{\bm{W}} &= \bm{V}_\mathrm{s}\diag{\frac{s_p^2}{s_p^2 + \sigma^2}}\trans{\bm{V}_\mathrm{s}}\bm{W} + \sigma\bm{V}_\mathrm{s}\diag{\frac{s_p}{s_p^2 + \sigma^2}}\trans{\bm{U}_\mathrm{s}}\bm{\epsilon},
\end{align}
which is the result given in \eqref{eq:latent_mean_svd}.

If we now let \(\sigma^2 \to 0\), we see that each \(\frac{s_p^2}{s_p^2 + \sigma^2} \to 1\). Moreover, when \(\bm{L}^{(W)}\) (or, equivalently, when \(\bm{L}\)) has full column rank, \(\bm{V}_\mathrm{s}\) is square and orthogonal, and so \(\bm{V}_\mathrm{s}\trans{\bm{V}_\mathrm{s}} = \bm{I}_R\), and the first term in \eqref{eq:latent_mean_svd_apx} converges to \(\bm{W}\). For the second term, in order to establish convergence in \(L^2\), we will show that its expected squared norm converges to zero. For ease of notation, we introduce
\begin{subequations}
\begin{align}
	\bm{n}_\sigma &= \sigma\bm{V}_\mathrm{s}\diag{\frac{s_p}{s_p^2 + \sigma^2}}\trans{\bm{U}_\mathrm{s}}\bm{\epsilon}~\text{and} \\
	\bm{A}_\sigma &= \sigma\bm{V}_\mathrm{s}\diag{\frac{s_p}{s_p^2 + \sigma^2}}\trans{\bm{U}_\mathrm{s}},
\end{align}
\end{subequations}
such that
\begin{equation}
	\bm{n}_\sigma = \bm{A}_\sigma \bm{\epsilon}.
\end{equation}
Its squared norm is
\begin{equation}
	\left\lVert\bm{n}_\sigma\right\rVert_2^2 = \trans{\bm{\epsilon}}\trans{\bm{A}_\sigma} \bm{A}_\sigma \bm{\epsilon},
\end{equation}
and as \(\bm{\epsilon}\) has covariance \(\bm{I}_{2F}\),
\begin{subequations}
\begin{align}
	\expectation{\left\lVert\bm{n}_\sigma\right\rVert_2^2} &= \trace{\trans{\bm{A}_\sigma} \bm{A}_\sigma} \\
	&= \sigma^2 \trace{\bm{U}_\mathrm{s}\diag{\frac{s_p}{s_p^2 + \sigma^2}}\trans{\bm{V}_\mathrm{s}}\bm{V}_\mathrm{s}\diag{\frac{s_p}{s_p^2 + \sigma^2}}\trans{\bm{U}_\mathrm{s}}} \\
	&= \sigma^2 \trace{\bm{U}_\mathrm{s}\diag{\frac{s_p^2}{\left(s_p^2 + \sigma^2\right)^2}}\trans{\bm{U}_\mathrm{s}}} \\
	&= \sigma^2 \trace{\diag{\frac{s_p^2}{\left(s_p^2 + \sigma^2\right)^2}}\trans{\bm{U}_\mathrm{s}}\bm{U}_\mathrm{s}} \\
	&= \sigma^2 \trace{\diag{\frac{s_p^2}{\left(s_p^2 + \sigma^2\right)^2}}} \\
	&= \sigma^2 \sum_{p = 1}^P\frac{s_p^2}{\left(s_p^2 + \sigma^2\right)^2}.
\end{align}
\end{subequations}
We now see that this term disappears as \(\sigma^2 \to 0\), and so \(\hat{\bm{W}}\to\bm{W}\) in \(L^2\). Convergence in probability follows from convergence in \(L^2\). To show that this also implies convergence of the \gls{miso-gptf} estimate, we subtract \eqref{eq:module_whitened_representation} from \eqref{eq:mean_whitened_representation}, take the squared norm and expectation to get
\begin{align}
	\hat{\check{\bm{G}}}_i - \check{\bm{G}}_i &= \bm{\Gamma}_i^{1/2} \left(\hat{\bm{W}}_i - \bm{W}_i \right)~\text{and} \\
	\expectation{\left\lVert\hat{\check{\bm{G}}}_i - \check{\bm{G}}_i\right\rVert_2^2} &\leq \left\lVert\bm{\Gamma}_i^{1/2}\right\rVert_2^2 \expectation{\left\lVert\hat{\bm{W}}_i - \bm{W}_i \right\rVert_2^2}.\label{eq:convergence}
\end{align}
As we have shown, when \(\sigma^2 \to 0\), \(\hat{\bm{W}}_i\to\bm{W}_i\), and so \eqref{eq:convergence} then gives that \(\hat{\check{\bm{G}}}_i \to \check{\bm{G}}_i\) for \(i = 0, \dots, m\) in \(L^2\), and hence in probability.\qed
    \newpage
	\section{Simulation Parameters}\label{apx:sim_params}
\begin{table}[htb]
  \caption{Parameters for the power system model in Section~\ref{sec:power_system}.}
  \label{tab:power_system_params}
  \centering
  \begin{tabular}{lrll}
    \toprule
    Parameter & Value & Unit & Description \\
    \midrule
    \(M_1\) & 6.00 & \unit{\second} & Inertia coefficient of bus 1 \\
    \(M_2\) & 1.50 & \unit{\second} & Virtual inertia coefficient of bus 2 \\
    \(D_1\) & 1.20 & \unit{\perunit} & Damping coefficient of bus 1 \\
    \(D_2\) & 2.00 & \unit{\perunit} & Virtual damping coefficient of bus 2\\
    
    \(K_{12}\) & 3.00 & \unit{\per\second} & Synchronizing coefficient, buses 1 and 2 \\

    \(K_\mathrm{g}\) & 10.00 & \unit{\perunit} & Governor gain at bus 1 \\
    \(T_\mathrm{g}\) & 0.20 & \unit{\second} & Governor time constant at bus 1\\
    \(T_\mathrm{t}\) & 0.70 & \unit{\second} & Turbine time constant at bus 1\\

    \(K_\mathrm{d}\) & 1.20 & \unit{\perunit} & Controller gain at bus 2\\
    \(T_\mathrm{c}\) & 0.08 & \unit{\second} & Filter time constant at bus 2\\
    \bottomrule
  \end{tabular}
\end{table}

	\bibliographystyle{elsarticle-num}
\bibliography{mybib}
\end{document}